\documentclass[notitlepage,aps,pra,superscriptaddress,nofootinbib,10pt,showpacs]{revtex4-2}

\usepackage{amsmath}
\usepackage{amssymb}
\usepackage{graphicx}
\usepackage{hyperref}
\usepackage{color}
\usepackage[T1]{fontenc}
\usepackage[utf8]{inputenc}
\usepackage{amsthm}
\usepackage{bbm}
\usepackage{mathtools}
\usepackage{xspace}
\usepackage[margin=2.5cm]{geometry}

\usepackage{xcolor}

\hypersetup{
	colorlinks,
	linkcolor={black!30!blue},
	citecolor={red},
	urlcolor={blue}
}

\mathtoolsset{centercolon}

\usepackage{pgfplots}
\usepgfplotslibrary{patchplots}
\usepackage{tikz}
\usetikzlibrary{arrows}
\usetikzlibrary{decorations.pathreplacing}
\usetikzlibrary{decorations.markings}
\usetikzlibrary{shapes,arrows,positioning,patterns}
\usetikzlibrary{matrix}
\usetikzlibrary{fit}
\usetikzlibrary{fadings}

\newtheorem{thm}{Theorem}[section]

\newtheorem{lem}[thm]{Lemma}
\newtheorem{prop}[thm]{Proposition}
\newtheorem{cor}[thm]{Corollary}
\newtheorem{remark}[thm]{Remark}
\newtheorem{example}[thm]{Example}

\newtheorem{schol}[thm]{Scholium}

\usepackage{xstring}
\DeclareRobustCommand{\Qref}[1]{\IfBeginWith{#1}{def}{Def.\;\ref{#1}}{}%
	\IfBeginWith{#1}{sec}{Sect.\;\ref{#1}}{}%
	\IfBeginWith{#1}{fig}{Fig.\;\ref{#1}}{}%
	\IfBeginWith{#1}{tab}{Table\;\ref{#1}}{}%
	\IfBeginWith{#1}{Tab}{Table\;\ref{#1}}{}%
	\IfBeginWith{#1}{thm}{Thm.\;\ref{#1}}{}%
	\IfBeginWith{#1}{lem}{Lem.\;\ref{#1}}{}%
	\IfBeginWith{#1}{cor}{Cor.\;\ref{#1}}{}%
	\IfBeginWith{#1}{prop}{Prop.\;\ref{#1}}{}%
	\IfBeginWith{#1}{conv}{Conv.\;\ref{#1}}{}%
	\IfBeginWith{#1}{schol}{Schol.\;\eqref{#1}}{}%
    \IfBeginWith{#1}{rm}{Rem.\;\eqref{#1}}{}%
}

\def\tr{\operatorname{tr}}
\def\idty{{\mathbbm 1}}

\def\Cx{{\mathbb C}}
\def\Ir{{\mathbb Z}}

\def\norm #1{\Vert #1\Vert}
\def\mod{{\mathop{\rm mod}\nolimits}}
\def\ind{{\mathop{\rm ind}\nolimits}\,}

\def\ket #1{\vert#1\rangle}
\def\braket#1#2{\langle #1,#2\rangle}
\def\ketbra #1#2{{\vert#1\rangle\langle#2\vert}}

\def\abs#1{\vert#1\vert}
\def\rank{{\mathop{\rm rank}\nolimits}\,}
\def\im{\Im m}

\def\BB{{\mathcal B}}
\def\HH{{\mathcal H}}
\def\KK{{\mathcal K}}

\makeatletter
\newcommand{\raisemath}[1]{\mathpalette{\raisem@th{#1}}}
\newcommand{\raisem@th}[3]{\raisebox{#1}{$#2#3$}}
\makeatother

\def\six{\mathop{\mathrm s\mkern-1mu\mathrm i}\nolimits}
\DeclareRobustCommand{\sixnp}{\mathop{\mathrm s\mkern-1mu \textup \i}\nolimits}
\DeclareRobustCommand{\sixR}{\mathop{\hbox{\raisebox{.45em}{$\scriptstyle\rightharpoonup$}}\hspace{-.9em}\sixnp}}
\def\sixL{\mathop{\hbox{\raisebox{.45em}{$\scriptstyle\leftharpoonup$}}\hspace{-.9em}\sixnp}}

\def\sixrel#1:#2{\six(#1{:}#2)}

\def\Pl#1{P_{<#1}}
\def\Pg#1{P_{\geq#1}}
\def\ch{\gamma}

\newcommand{\quasinv}[1]{{#1}^I}
\newcommand{\RNum}[1]{\uppercase\expandafter{\romannumeral #1\relax}}
\def\DK{\dim\ker}
\def\dk{\dim\ker}
\def\adj{^*}
\def\fred#1{\ind[#1]}

\DeclareMathOperator\indFRnp{\mbox{$\ind\hspace{-1.3em}\raisebox{.55em}{$\scriptstyle\rightharpoonup$}\hspace{.5em}$}}
\def\fredP#1{\indFRnp[#1]}

\def\AA{\mathcal A}
\def\AAloc{{\mathcal A}_{\mathrm{loc}}}

\begin{document}

\title{Chiral Floquet systems and quantum walks at half period}
\author{C. Cedzich}
\affiliation{Laboratoire de Recherche en Informatique (LRI), Universit\'{e} Paris Sud, CNRS, Centrale Sup\'{e}lec, \\B\^{a}t 650, Rue Noetzlin, 91190 Gif-sur-Yvette, France}
\author{T. Geib}
\affiliation{Institut f\"ur Theoretische Physik, Leibniz Universit\"at Hannover, Appelstr. 2, 30167 Hannover, Germany}
\author{A.~H. Werner}
\affiliation{{QMATH}, Department of Mathematical Sciences, University of Copenhagen, Universitetsparken 5, 2100 Copenhagen, Denmark,}
\affiliation{{NBIA}, Niels Bohr Institute, University of Copenhagen, Denmark}
\author{R.~F. Werner}
\affiliation{Institut f\"ur Theoretische Physik, Leibniz Universit\"at Hannover, Appelstr. 2, 30167 Hannover, Germany}

\begin{abstract}
	We classify chiral symmetric periodically driven quantum systems on a one-dimensional lattice. The driving process is local, can be continuous or discrete in time, and we assume a gap condition for the corresponding Floquet operator. The analysis is in terms of the unitary operator at a half-period, the half-step operator. We give a complete classification of the connected classes of half-step operators in terms of five integer indices. On the basis of these indices it can be decided whether the half-step operator can be obtained from a continuous Hamiltonian driving, or not. The half-step operator determines two Floquet operators, obtained by starting the driving at zero or at half period, respectively. These are called timeframes and are chiral symmetric quantum walks. Conversely, we show under which conditions two chiral symmetric walks determine a common half-step operator. Moreover, we clarify the connection between the classification of half-step operators and the corresponding quantum walks. Within this theory we prove bulk-edge correspondence and show that a second timeframe allows to distinguish between symmetry protected edge states at $+1$ and $-1$ which is not possible for a single timeframe.
\end{abstract}

\maketitle

\section{Introduction}

Topological classifications of quantum systems in the presence of symmetries provide novel phases of matter beyond the Landau theory of symmetry breaking. Prominent examples in physics are the distinction between ordinary and topological insulators \cite{SchuBaBook,kitaevPeriodic,KaneMele1,KaneMele2,HasanKaneReview,ZhangReview} or the quantization of the Hall conductance in two-dimensional samples \cite{Quantum_Hall_Monograph}, but the theory extends to all dimensions and all symmetry classes of the tenfold way \cite{Altland-Zirnbauer}. The key feature of these novel phases is their exceptional stability against all perturbations which preserve basic elements of the theory, typically a symmetry, a gap condition, and a standing assumption on the locality of interactions.
This is mathematically expressed by a non-trivial topological classification up to homotopy. From a physical point of view, an important consequence of this stability is the so-called bulk-boundary correspondence which allows to predict phenomena occurring near the boundary between two systems that belong to distinct topological classes.

Recently, the topological classification of lattice systems with discrete symmetries has become one of the cornerstones of a theory of topological quantum matter. Model systems beyond static Hamiltonians, in which topological phases are still relevant, come in two settings. On one hand they are Floquet systems with continuous time periodic driving \cite{RunderBergLevin,graf2018bulk,sadel2017topological,roy2017periodic,liu2018chiral,kitagawa2010topological,Asbo3,carpentier2015topological,fruchart2016complex} and, on the other, discrete time quantum walks \cite{ShortVersion,LongVersion,UsOnTI,WeAreSchur,kita1,kita2,Asbo1,Asbo2}. These classes are closely related, because the Floquet evolution operator for a full period is a walk and, conversely many walks can be generated by an ``effective Hamiltonian'' or as a sequence of pulses in continuous time \cite{delplace2017phase}. This correspondence was fruitful and led to the first results on topological properties of quantum walks \cite{kita1,kita2,Asbo1,Asbo2}.

However, the two setups are fundamentally different in what is naturally required for a system to ``satisfy a symmetry'': In discrete time, the symmetries act directly on the walk unitary itself and map it either to itself or to its adjoint. In contrast, in periodically driven systems it is natural to impose the symmetries not only at the endpoint of one period but for the whole generating process. When the symmetry involves a time inversion  this connects the driving Hamiltonian at time $t$ not only with the one a period later,  at $t+T$, but also the one at $-t$. For chiral symmetry, these conceptual differences lead to two distinct topological classifications. In fact, the present work grew out of the project of understanding the apparent contradictions between the classification of chiral walks in \cite{LongVersion} and the results in \cite{Asbo2}. It turns out that the key structure for understanding these connections is intermediate between the two settings: In the setting of periodically driven systems one notices that the driving process in the second half of the driving period is completely determined by the first half. Therefore, all relevant information is already contained in the unitary operator (denoted $F$ in the sequel) one obtains at half a period. From the point of the discrete time walk $F$ appears as a kind of square root of the walk operator. It is not exactly a root, however, because the two half periods are not equal, but only connected by symmetry and taking an adjoint. Therefore, the two halves usually do not commute, and so we obtain not one walk, but two walks, depending on whether we set the time origin to $0$ or to $T/2$. Both are stroboscopic Floquet samples of the same process, and are called the two timeframes associated with the process. Both timeframed walks are obtained from $F$ by multiplication with a chirally transformed $F$, so it is natural to take the operators $F$ as the starting point. This even works if $F$ is not obtained by driving with a local Hamiltonian, and this additional possibility covers some further examples from the literature.

The aim of this paper is to give a complete topological classification of the structures just described. In this endeavour we build on our earlier classification of walk unitaries with symmetry
\cite{ShortVersion,LongVersion,UsOnTI,WeAreSchur} on one hand, and the work emphasizing two timeframes in \cite{Asbo2,Asbo4,obuse2015unveiling} on the other. Our classification is expressed in terms of five integer indices associated with each operator $F$. These indices also have a direct implication for the appearance of edge states in bulk-boundary correspondence, and allow to decide whether $F$ can be obtained from a local Hamiltonian driving or not.

This paper is organized as follows: In the next section we give a detailed account of different meanings in which classifications of chiral symmetric quantum walks were discussed, lay down the setting we study, and collect the main results of this work. Moreover, to make the present paper as self-contained as possible we summarize the main points of the topological classification of quantum walks. In \Qref{sec:halfstepwalks} we express the symmetry indices of the two timeframes of chiral symmetric walks in terms of Fredholm indices of the chiral blocks of the half-step walk. Also, we explicitly construct such a half-step walk for every chiral symmetric walk. Taking locality requirements into account in \Qref{sec:locality} gives a total of ten indices for the half-step walk. However not all of them are independent, which allows us to reduce this number to a set of five independent indices. In \Qref{sec:complete_indices} we prove our main classification result (\Qref{thm:complete}) which states, that this set of five indices provides complete classification of half-step walks up to homotopies and compact perturbations. Subsequently, in \Qref{sec:time_framed_walks} we clarify the relation between the classification of two timeframes via these half-step indices and the existing classification of single timeframes. 
\section{Setting and main results}
\subsection{Walks and Timeframes}\label{sec:walks_protocols}
We consider quantum walks with chiral symmetry in one dimension. This
terminology has been used with two distinct meanings in the literature, and our
aim here is to clarify the connections. The Hilbert space in both cases is of the form
\begin{equation}\label{eq:H}
  \HH=\bigoplus_{x\in\Ir}\HH_x,
\end{equation}
where the variable $x$ labels a point on the lattice, or a ``cell'', and the local degrees of freedom are given by a finite-dimensional Hilbert space $\HH_x$.
The cell dimensions need not be uniform, which allows the joining of two bulk systems with very different characteristics.
For every $a\in\Ir$ we denote by $\Pg a$ and $\Pl a=\idty-\Pg a$ the orthogonal projections onto the \textbf{half-chains} $\HH_{\geq a}=\bigoplus_{x\geq a}\HH_x$ and $\HH_{<a}=\bigoplus_{x< a}\HH_x$, respectively.
We consider a \textbf{chiral symmetry} acting by a unitary operator $\ch$ with $\ch^2=\idty$, leaving each cell invariant.
We assume that the eigenspaces of $\gamma$ in $\HH_x$ have the same multiplicity, so each $\HH_x$ is even dimensional.

Next, we have to describe how the chiral symmetry operators $\ch$ constrains the class of symmetric operators that will be part of our classification. We can immediately distinguish two extreme cases here:

\begin{enumerate}
\item[($W$)] In the {\bf one-step scenario} a (quantum) ``walk'' is defined as a unitary operator $W$ on $\HH$.
\item[($H$)] In the {\bf Floquet scenario} a ``walk'' is the Floquet operator of a process in continuous time, driven
by a (bounded) time dependent local Hamiltonian $H(t)$, which is periodic in time, say, $H(t+1)=H(t)$, where we have fixed the period to be $1$ by a choice of units.
\end{enumerate}
The connection between the two is obviously,
\begin{equation}\label{TexpW}
  W={\mathcal T}e^{-i\int_0^1 dt H(t)},
\end{equation}
which is the solution $W(t)|_{t=1}$ of the time dependent Schr\"odinger equation
$\partial_t W(t)=-iH(t)\,W(t)$ with initial condition $W(0)=\idty$.
We will often call this scenario a \textbf{continuous driving}, which refers to the norm continuity of $W(t)$ in $t$. Note, that $H(t)$ does not have to be continuous in $t$ for this condition to be met (see \Qref{prop:inx0} (3) for a precise formulation).
Some of the typical assumptions are readily transferred between these settings.
For example, a {\it local} Hamiltonian is one for which matrix elements $\braket{\phi_x}{H\psi_y}$ decay sufficiently fast with the separation $\abs{x-y}$, which implies a similar condition for $W$ (Precise conditions will be specified later).

We are concerned here with the chiral symmetry, which in terms of $W$ is just the condition
\begin{equation}\label{Wchiral}
    \ch W\ch\adj=W\adj.
\end{equation}
In in the one-step scenario $W$, where we do not have the information on how the operator was implemented, that is all we can say. Indeed, in experiments \cite{Karski2009,Regensburger2011,Schmitz2009,Cardanoe2015} the driving is often not implemented as in \eqref{TexpW}, i.e., not by a Hamiltonian in $\HH$, but temporarily uses additional degrees of freedom of the system, for which the chiral symmetry might not even be defined. Setting ($W$) is then clearly the appropriate one.
On the other hand, for a continuous driving process in $\HH$ it is natural to impose a more detailed condition, namely
\begin{equation}\label{chiraldrive}
  \ch H(t)\ch\adj=-H(-t).
\end{equation}
This readily implies \eqref{Wchiral}. By this condition, together with periodicity, we get $H(1-t)=H(-t)=-\ch H(t)\ch\adj$, so the Hamiltonian for $t\in[1/2,1]$ can be computed from the values for $t\in[0,1/2]$.
For the time ordered exponentials this gives
\begin{equation}\label{TexpF}
  F={\mathcal T}e^{-i\int_0^{1/2} dt H(t)}, \qquad {\mathcal T}e^{-i\int_{1/2}^1 dt H(t)}=\ch F\adj \ch.
\end{equation}
Hence the Floquet operators $W$ and $W'$ for time periods starting at integer times and half-integer times are
\begin{equation}\label{eq:timeframes}
  W=\ch F\adj\ch F    \quad\text{and}\quad W'=F \ch F\adj\ch,
\end{equation}
and are called the \textbf{timeframes} of the process.
They describe the same periodically driven process, but with the origin shifted by half a period.
It is important to note that, while both unitary operators are continuously connected by the Floquet operators for other choices $\tau$ as the origin, i.e., the time ordered exponentials for the intervals $[\tau,\tau+1]$ only the two choices $\tau\equiv0,1/2\ \mod\Ir$ satisfy the chiral symmetry. Hence in a classification based on the deformation of {\it chiral symmetric walks}, $W$ and $W'$ may be in different classes.

Our theory will primarily be based on the relations \eqref{eq:timeframes}, so we take the unitary operator $F$ as an object in its own right. 
We therefore introduce a third scenario besides ($W$) and ($H$), namely
\begin{itemize}
\item[($F$)] In the \textbf{half-step scenario}, a ``walk'' is determined by a unitary half-step operator $F$ via \eqref{eq:timeframes}.
\end{itemize}
Note that there is no chiral symmetry condition left to be imposed on $F$: Chirality is in this case encoded in the way the timeframes are built from $F$. 
These scenarios and the relations connecting them are summarized in the diagram
\begin{equation}
\label{diag}
	\begin{tikzpicture}[baseline=(current  bounding  box.center),>=stealth]
		\node (H) at (0,1) {$H(t)$};
		\node (F) at (0,0) {$F$};
		\node (W) at (-1,-1) {$W$};
		\node (W') at (1,-1) {$W'$};
		
		\draw[->] (H) -- (F);
		\draw[->] (F) -- (W);
		\draw[->] (F) --(W');
		\node at (2,-.5) {\small see \eqref{eq:timeframes}};
		\node at (2,.5) {\small see \eqref{TexpF}};
	\end{tikzpicture}
\end{equation}

Adopting the scenario ($F$) allows us to effortlessly include cases such as the so-called split-step walk, which cannot be realized by any continuous driving. More generally we include {\bf discrete driving} sequences $W=U_{2n}\cdots U_1$ by a sequence of unitary operators such that the discrete analogue $\gamma U_i \gamma = U_{2n-i+1}^*$ of \eqref{chiraldrive} holds. Then $F=U_n\cdots U_{1}$. In the continuous driving scenario there would be an additional divisibility constraint, namely that the steps $U_i$ can be taken to be close to the identity.

Throughout we  impose the following additional constraints:
\begin{itemize}
\item[\quad]\hskip-20pt{\bf Chirality} as given by Eqs. \eqref{Wchiral} and \eqref{chiraldrive}. There is no chirality condition on $F$:
 In this scenario the symmetry is encoded in the way the overall process is built from $F$.
\item[\quad]\hskip-20pt{\bf Essential Locality} The matrix elements of $H, F$ or $W$ between distant sites should decay. We encode this in the condition that, for some (hence for all) $a\in\Ir$, the half-line projection $\Pg a$ essentially splits the system, in the sense that $\Pg aH\Pl a$ is a compact operator, and similarly for $F$ or $W$.
\item[\quad]\hskip-20pt{\bf Gap condition}, the spectrum of $W$ has essential gaps at the symmetry invariant points $\pm1$. This means that, in an interval around these points, the spectrum of $W$ contains only eigenvalues of finite multiplicity. Since $W'$ is unitarily equivalent to $W$, the same then holds for $W'$.  For $H$ and $F$ the condition is implicitly stated in terms of $W$.
\end{itemize}
\noindent Of course, the locality condition allows strictly local operators, for which distant matrix elements are exactly zero, but this condition would not survive exponentiation. In contrast, the essentially local operators form a C*-algebra, and time ordered exponentiation does not lead out of it \cite{LongVersion,UsOnTI}. For the gap condition note that \eqref{Wchiral} implies that the spectrum of $W$ is invariant under complex conjugation. The points $\pm1$ thus play a distinguished role.

Now since each arrow in the diagram \eqref{diag} represents a norm continuous function, any continuous deformation of the driving process translates to a deformation of $F,W,W'$. Therefore, a topological classification of walks in the one-step scenario automatically gives a partial classification of driving Hamiltonians or half-step unitaries. The interesting questions therefore concern the opposite direction:
\begin{enumerate}
\item Inversion of arrows: Given a walk $W$ in the one step scenario, can we find an $F$ according to \eqref{eq:timeframes}? When does a pair $(W,W')$ arise from $F$? Given $F$, can we find a driving Hamiltonian?
\item Classification: In each scenario, classify the walks up to deformations respecting all three of the above conditions. Express the classification in terms of discrete indices, and show that equality of indices is equivalent to the existence of a deformation (completeness of indices).
\item No empty classes: Show that each index combination can be realized by suitable objects.
\item Given the topological classes of $W$ and $W'$, is the class of $F$ determined?
\end{enumerate}

\subsection{Review of the classification in the one-step scenario}\label{sec:review}
Questions 2 and 3 have been answered in \cite{LongVersion,ShortVersion,UsOnTI,WeAreSchur}, not just for chiral symmetry but for all symmetry types of the tenfold way \cite{Altland-Zirnbauer}.
So we can just specialize to the chiral case (conventionally labelled {\mbox{{\sf A}I\kern-1ptI\kern-1pt I}}). Our terminology and notation ``$\six$'' for ``symmetry index'' was to emphasize in this wider context that the quantity takes values in an abelian group depending on the symmetry type. In the chiral case this group is $\Ir$.

According to \cite{LongVersion}, the classification of one-step walks is by triples of integer indices
\begin{equation}\label{long-indices}
  \bigl(\six(W),\sixR(W),\six_+(W)\bigr),
\end{equation}
which exactly label the homotopy classes. We will presently define these, but first offer an intuitive interpretation. The index $\sixR(W)$ describes the structure far to the right, and can be determined from the walk projected to any subspace $\Pg a\HH$. Similarly, $\sixL(W)=\six(W)-\sixR(W)$ relates to the far left. The index $\six_+(W)$ classifies the symmetry protected eigenspace of $W$ at $+1$ and similarly $\six_-(W)=\six(W)-\six_+(W)$ the one at $-1$. $\six(W)$ acts as the connecting bit between these two types of indices. In the translation invariant case the index triples are of the form $(0,n,0)$, and each $n$ is given by a winding number determined from the band structure. The possibility of having $\six(W)\neq0$ is crucial for describing a crossover of two bulk systems with different topological characteristics. In such systems, topologically protected eigenvalues appear in the spectral gap at $\pm1$, which is known as bulk-boundary correspondence. It is expressed by the sum relation
\begin{equation}\label{WindexSums}
\six(W)=\six_-(W)+\six_+(W)=\sixR(W)+\sixL(W).
\end{equation}
Here $\abs{\six_\pm(W)}$ is a lower bound to the dimension of the eigenspace of $W$ at $\pm1$. Hence if $\six(W)\neq0$, \eqref{WindexSums} implies some topologically protected eigenvalues in the gap.
Since $\six_\pm(W)$ are individually invariant under deformation, topological protection also holds for these eigenvalues separately. A new phenomenon in the classification of unitaries (as opposed to Hamiltonians, i.e., continuous time systems) is the possibility of ``non-gentle'' local perturbations: These are modifications of a walk, for example only on finitely many cells, which cannot be done along a continuous path (preserving all the defining conditions). Then the index $\six_+(W)$ exactly classifies this new possibility, i.e., the set of compact perturbations up to homotopy. In contrast, the right/left indices are invariant also under compact perturbations. For the bulk-boundary correspondence this means that engineering a crossover between two topologically distinct half-infinite bulks, which is always possible, can be done in topologically distinct ways. By definition, these are linked by a compact perturbation, and are hence labelled by the third index in \eqref{long-indices}.

It remains to define the index quantities in \eqref{long-indices}. The pairs $\sixR,\sixL$ and $\six_\pm$ depend differently on the underlying structures, and since the same will be true for the half-step classification work, we will point out these differences here. For $\six_\pm(W)$ no local structure is required, only the chiral symmetry and an essential gap. By symmetry the eigenspaces $\HH_\pm=\ker(W\mp\idty)$ at $\pm1$ are invariant under $\ch$, and by the gap condition they are finite dimensional. Hence we can define
\begin{equation}\label{sixpm}
  \six_\pm(W)=\tr_{\HH_\pm}\ch,
\end{equation}
which one might call the chiral imbalance of the eigenspaces. It is deformation invariant, because according to perturbation theory a branching of this eigenvalue will lead to a complex conjugate pair of non-real eigenvalues, so that in the joint eigenspace a gapped chiral unitary exists. This implies vanishing trace of $\ch$ in the subspace, and hence that $\tr\ch$ is unchanged in the eigenspace. With $\im(W)=(W-W^*)/(2i)$ the overall index $\six=\six_++\six_-$ can be expressed as 
\begin{equation}
	\six(W)=\tr_{\ker\im(W)}\ch.
\end{equation}
	
The indices $\sixR,\sixL$ obviously do depend on the local structure. One would like to define them in terms of some version of the walk definable to the far right and left. A possible approach is to perform a continuous deformation which decouples the walk to the form $W_L\oplus W_R$, and then use the $\pm$ indices for the parts. This works \cite{ShortVersion}, but introduces a dependence on the decoupling chosen, which makes the indices $\six_\pm(W_R)$ individually ill defined (due to non-gentle perturbations), leaving only their sum $\sixR(W)=\six(W_R)$ as a possible invariant. Of course, this has then to be shown to be independent of the decoupling process. A better approach \cite{ShortVersion} takes its starting point from the operator $PWP$, where $P=\Pg0$ is the right half projection. This is {\it not unitary} but, due to the essential locality condition, ``essentially unitary'' meaning that $(PWP)^*(PWP)-P$  and $(PWP)(PWP)^*-P$ are compact, so $PWP$ is Fredholm. This suffices to apply index theory to the imaginary part $P(W-W\adj)P/(2i)$, and $\sixR(W)$ is defined as $\tr\ch$ on the $0$-eigenspace  of this operator. This approach makes clear why the distinction between the $\pm1$ eigenspaces is wiped out for the left/right indices, but stability under compact perturbations is gained.

For the current paper we follow the path by considering also a projected version $PFP$ of the half-step operator. Thus some propositions (not all) are expressly stated for {\it essentially unitary} $F$, not because we think this is an interesting generalization in itself, but because such statements can be applied to the projected version $PFP$.

\subsection{Results}\label{sec:results}
We are requiring of $F$ that the timeframes $W,W'$ are essentially gapped. As a statement about the spectrum of a product of operators on an infinite dimensional Hilbert space, this is a rather indirect condition, and one would like a more direct one. It turns out that an extremely elegant characterization is possible. To this end we decompose the Hilbert space into the eigenspaces of the chiral unitary $\gamma$, and write $F$ in block matrix form. The four \textbf{chiral blocks} in this decomposition will denoted below by $(A,B,C,D)$:

\begin{align}\label{Fblock}
	\ch=\begin{pmatrix}\idty & 0\\0 & -\idty\end{pmatrix}\qquad&\text{and}\qquad F =\begin{pmatrix} A & B\\ C & D \end{pmatrix},\\
\intertext{Expressed in terms of  the chiral blocks $W$ and $W'$ from \eqref{eq:timeframes} are of the form}
	W=\begin{pmatrix}
		A^*A-C^*C & A^*B-C^*D\\D^*C-B^*A & D^*D-B^*B
	\end{pmatrix}\qquad&\text{and}\qquad
	W'=\begin{pmatrix}
		AA^*-BB^* & BD^*-AC^*\\CA^*-DB^* & DD^*-CC^*
	\end{pmatrix}.\label{eq:walks}
\end{align}

Then $W$ and/or $W'$ are gapped iff each of the four blocks of $F$ is a Fredholm operator (see \Qref{lem:sixFeu}).
A Fredholm operator $X$ comes with an integer index, denoted by $\fred X$, which completely describes the homotopy classes of such operators. The indices of the four chiral blocks of $F$ are thus candidates for a classification of half-step walks. However, these indices do not distinguish between the right and the left halves of the system, i.e., they are independent of any locality constraints. This is a virtue in view of potential applications to higher dimensional systems, and so we collected all results {\it not} using the one-dimensional underlying lattice in \Qref{sec:halfstepwalks}. On the other hand, the classification of half-step walks must map onto the classes of  one-step walks, which does reflect the locality. In order to introduce locality in \Qref{sec:locality} we fix a half space projection $P=\Pg a$ (the cut point $a$ will be irrelevant), and define a ``directed'' version of the Fredholm index by $\fredP{X}=\fred{PXP|_{P\HH}}$. With these notations we can state our central classification result, answering Question 2 and 3 above:

\begin{thm}\label{thm:complete}
Consider the class of half-step walks satisfying the conditions stated in \Qref{sec:walks_protocols}, and associate with each half-step walks $F$ with block decomposition \eqref{Fblock} the quintuple
\begin{equation}\label{allinx}
	\left(\fred{A},\fred{B},\fredP{A},\fredP{B},\fredP{C}\right)\in\Ir^5.
\end{equation}
Then the indices are invariant under norm continuous deformations and compact perturbations. Conversely, any two half-step walks with the same indices can be deformed into each other. Moreover, all index tuples can be realized by suitable $F$.
\end{thm}

The conspicuously absent indices are functions of these five:
$\fred C=-\fred B$, $\fred D=-\fred A$, and $\fredP D=\fredP{B}+\fredP{C}-\fredP{A}$. We prove \Qref{thm:complete} in \Qref{sec:complete_indices}: For the independence and invariance of the five indices see \Qref{sec:independent_set}. The realizability of all index combinations is discussed in \Qref{sec:generating_example}, where we define a \textbf{generating example}, which allows us to realize every index combination in $\Ir^5$. Completeness is proved in \Qref{sec:completeness_proof}.

As expected, the indices of the timeframes are functions of these:
\begin{eqnarray}
  \bigl(\six(W),\sixR(W),\six_+(W)\bigr)    &=& \bigl(\fred{A}-\fred{B},\,\fredP{A}-\fredP{B},\,-\fred{B}\bigr) \\
  \bigl(\six(W'),\sixR(W'),\six_+(W')\bigr) &=& \bigl(-\fred{A}-\fred{B},\,\fredP{C}-\fredP{A},\,-\fred{B}\bigr)
\end{eqnarray}
See \Qref{lem:sixFeu}, \Qref{sixFu} and the discussion around \eqref{eq:sixR} and \eqref{eq:sixRprime} for proofs of these formulas.

Note that $\six_\pm(W)=\pm\six_\pm(W')$ which is thus a necessary condition for these to be timeframes of the same $F$. Moreover, the index tuples of $W$ and $W'$ are not sufficient to reconstruct the index tuple of $F$. In \Qref{sec:time_framed_walks} we provide an example, in which this can be traced to an ambiguity in the reconstruction of $F$. The missing quantity is the right Fredholm index
\begin{equation}\label{totalFred}
    \fredP F=\fredP{B}+\fredP{C}=\fredP{A}+\fredP{D}.
\end{equation}
of $F$ itself. This quantity, also known as the Kitaev flow \cite{Kitaev}, or just index \cite{old_index_paper} classifies the connected components of essentially local unitary operators (no gap or symmetry required). Since, by definition any Floquet operator is continuously connected to the identity, the connection to the Floquet scenario implies $\fredP F=0$. The following Theorem summarizes our results on ``inversion of arrows'', i.e.\ Question 1 and 4 above. We assume all objects to satisfy the standing assumptions of  \Qref{sec:walks_protocols}.

\begin{thm}\label{thm:allinverse}\strut\\\vspace{-2ex}
\begin{enumerate}
\item[(1)] For any chiral one-step walk $W$ a half-step walk $F$ satisfying \eqref{eq:timeframes} exists.
\item[(2)] A pair $(W,W')$ arises from some half-step walk $F$ via \eqref{eq:timeframes} iff $\six_\pm(W)=\pm\six_\pm(W')$ and there is a unitary operator $U$ such that $W'=UWU^*$ and $U\gamma_\KK U^*=\gamma'_\KK$, where $\gamma_\KK,\gamma'_\KK$ is the chiral symmetry restricted to the complements of the $\pm1$ eigenspaces of $W,W'$, respectively.
\item[(3)] Two index triples $(\six,\sixR,\six_+)$ and $(\six',\sixR',\six_+')$ have timeframe representatives $W$ and $W'$ iff $\six_\pm=\pm\six_\pm$, i.e.\ this is the only constraint on the topological classes of $W$ and $W'$.
\item[(4)] A half-step walk $F$ arises from some continuous driving $H(t)$ iff $\fredP F=0$.
\end{enumerate}
\end{thm}
\noindent We prove the first three statements in \Qref{prop:Fexists}, \Qref{thm:twotimeframes} and \Qref{cor:(si+si')/2}. For the last statement, see \Qref{prop:inx0}. 
\section{The half-step operator}\label{sec:halfstepwalks}

In this section we begin the analysis of half-step operators $F$ as objects in their own right. It is clear from formula \eqref{eq:timeframes} that such an operator determines a pair of unitary operators $W$ and $W'$, and our standing assumption will be that these are chiral and essentially gapped in the sense of \Qref{sec:walks_protocols}. For the homotopy classification of such half-step operators we thus allow any deformation which preserves these properties. Since $W$ and $W'$ clearly depend continuously on $F$ this implies that the indices of the walks are homotopy invariants for $F$. But are there more? Thinking of $F$ as a more detailed description of the process, a finer description might be available for arbitrary $W$. Moreover, what are the conditions on $W$ and $W'$ for being related as in \eqref{eq:timeframes}?

We will begin with a setting, in which only the chiral symmetry is taken into account, i.e., essential locality is not yet imposed. All the statements in the current section will therefore be valid in any lattice dimension, or indeed without considering a local structure of any kind. Since some locality condition is usually part of the definition of quantum walks, we thus speak of unitary operators $W,W'$ and a half-space {\it operator}, reserving the term half-step {\it walk} for the setting with locality.
The indices of essentially gapped chiral unitaries are then reduced to $\six_\pm(W)$, and the right index $\sixR$ is not even defined. However, in preparation for the right/left indices we distinguish two kinds of results: One kind, e.g., \Qref{lem:sixFeu} and \Qref{lem:iF}, is valid for essentially unitary $F$, and these can later be applied to half-space projections of $F$. The other kind, e.g. \Qref{lem:sixFu}, requires exact unitarity and typically involves the eigenspaces at $\pm1$, assuming implicitly that the operator is normal and thus has a good spectral resolution.

In \Qref{sec:InxRel} we provide formulas for the indices of $W,W'$ in terms of Fredholm indices of the blocks of $F$. This will imply a necessary index relation for pairs $W,W'$, i.e., $\six_\pm(W)=\pm\six_\pm(W')$. The Fredholm Index of $F$ is likewise determined in a straightforward way by the indices of the blocks (\Qref{sec:FredF}). But does a half-step operator always exist? There are two issues here, discussed in \Qref{sec:existence}. We first show that {\it any} essentially gapped chiral $W$ can be written in terms of a half-step operator. For pairs $W,W'$ to be related by \eqref{eq:timeframes} we have already established an index constraint. Moreover, two unitary operators that can be written as $U_1U_2$ and $U_2U_1$ in terms of other unitaries are always unitarily equivalent. This is clearly applies to $W,W'$ in \eqref{eq:timeframes}, so $W$ and $W'$ have the same spectral data as well. In addition we will need that the intertwining unitary satisfies a condition involving the chiral symmetry. The resulting characterization of pairs $W,W'$ is given in \Qref{thm:twotimeframes}.

Essential locality is then added to the picture in \Qref{sec:locality}. We will impose this as a condition on $F$, which certainly makes sense for an $F$ obtained by continuous driving, but also for more general protocols. Since $\ch$ acts locally in each cell, $W$ and $W'$ are then also essentially local. For the existence results the inclusion of locality also works: The construction given in \Qref{sec:existence} naturally preserves essential locality. So all chiral walks can be written in terms of a half-step walk as well.

\subsection{Index relations for the half-step operator}\label{sec:InxRel}

Recall that a bounded operator $X$ is called essentially invertible or \textbf{Fredholm} if there is a (bounded) operator $\quasinv{X}$ such that $X\quasinv{X}-\idty$ and $\quasinv{X}X-\idty$ are compact. If no further restrictions are assumed, Fredholm operators are completely classified by the integer-valued \textbf{Fredholm index}
\begin{equation}
  X\mapsto\fred{X}=\dim \ker(X)-\dim \ker(X^*),
\end{equation}
which satisfies $\fred{X^*}=-\fred{X}$ and $\fred{X_1X_2}=\fred{X_1}+\fred{X_2}$ for $X_1,X_2$ Fredholm.
The Fredholm index is constant on norm continuous paths of Fredholm operators. Moreover, it is invariant under compact perturbations, i.e.\ for a compact operator $K$, $X+K$ stays Fredholm and we get $\fred{X+K}=\fred{X}$.

The following lemma relates the gap condition of $W$ and $W'$ to the essential invertibility of the chiral blocks of the half-step operator $F$. Moreover, it provides an expression of the symmetry indices of $W$ and $W'$ in terms of Fredholm indices of the chiral blocks $A,B,C$ and $D$.

\begin{lem}\label{lem:sixFeu}
Let $F$ be an essentially unitary operator in chiral block representation
\eqref{Fblock}, and let $W,W'$ be the essentially unitary operators determined from $F$ via \eqref{eq:timeframes}.
Then the following are equivalent:
\begin{itemize}
\item[(1)] $W$ is essentially gapped.
\item[(2)] $W'$ is essentially gapped.
\item[(3)] $A$, $B$, $C$, $D$ are Fredholm operators.
\end{itemize}
In this case
\begin{align}
\six(W) &=\fred{A}-\fred{B}\quad=\quad\fred{C}-\fred{D}\label{eq:six},\\
\six(W') &=\fred{C}-\fred{A}\quad=\quad\fred{D}-\fred{B}\label{eq:sixprime}.
\end{align}

\end{lem}

\begin{proof}
(1)$\Leftrightarrow$(3):
For the first part we can replace each operator by its image in the Calkin
algebra. Then $F$, $W$ and $W'$ are exactly unitary. Item (1) means that $\pm1$ are not in the
spectrum of $W$, i.e.,
$\ch F^*\ch F\mp\idty=\ch F^*\ch(F\mp\ch F\ch)$ is invertible in the Calkin
algebra. Since $\ch F^*\ch $ is unitary, this is saying that $( F\mp\ch F\ch)$
are invertible. But these are just the versions of \eqref{Fblock} with either
the diagonal or the off-diagonal blocks set equal to zero. Clearly such an
operator is invertible iff its non-zero blocks are. That is $A,B,C,D$ must be
invertible in the Calkin algebra, that is Fredholm operators.

\noindent(1)$\Leftrightarrow$(2): In the Calkin algebra $W'=FWF^*$ and the
essential spectra of $W$ and $W'$ are the same.

\noindent To prove the index formulas we consider the null space of
\begin{equation}\label{imW}
  \im W=(W-W^*)/(2i)
    =\frac1i\begin{pmatrix} 0 & A^*B-C^*D\\D^*C-B^*A & 0\end{pmatrix}
    =:\begin{pmatrix} 0 & X\\-X^*&0 \end{pmatrix}.
\end{equation}
Even when $W$ is not exactly unitary, and not diagonalizable, this operator is
hermitian, and its eigenspaces are well defined. Its kernel consists of the
vectors $(\psi_+,\psi_-)$ with $\psi_+\in\ker X^*$ and $\psi_-\in \ker X$, so
$\six(\im W)=\DK {X^*}-\DK X=-\fred{X}$, see \cite[Lemma 3.8]{UsOnTI}.
Since by essential unitarity
$A^*B+C^*D$ is compact, the Fredholm operator $2i A^*B$ is a compact
perturbation of $X$, and has the same index as $X$, namely
$\fred{A^*B}=\fred{B}-\fred{A}$.
\end{proof}

The statement of this lemma remains true when all ``essentially''s are dropped, i.e. in the exactly unitary setting the exact gappedness of $W$ or $W'$ is equivalent to the exact invertibility of the chiral blocks. In contrast, the following lemma requires exact unitarity and allows us to express $\six_+$ and $\six_-$ in terms of Fredholm indices of the chiral blocks of $F$:

\begin{lem}\label{lem:sixFu}
Let $F$ be a unitary operator in chiral block representation
\eqref{Fblock}, and let the corresponding unitary operators $W,W'$ from \eqref{eq:timeframes} be
essentially gapped.
Then, the symmetry indices can be defined separately for the eigenspaces at
$\pm1$, and we get
\begin{align}\label{pmindices}
	\six_+(W) &= \fred{C}=-\fred{B} =\six_+(W')\\
	\six_-(W) &= \fred{A}=-\fred{D}=-\six_-(W').
\end{align}
\end{lem}

\begin{proof}
The unitarity condition $F^*F=\idty$ implies that
\begin{equation}\label{eq:unitarity_of_F}
\begin{aligned}
  A^*A+C^*C &=\idty= B^*B+D^*D \\
  A^*B+C^*D&=0= B^*A+D^*C .
\end{aligned}
\end{equation}
Now suppose that $\psi\in\ker A$. Then $D^*C\psi=0$, i.e., $C\psi\in\ker D^*$.
In the same way we get $4$ inclusions and another $4$ from $FF^*=\idty$.
\begin{alignat*}{4}
  &C\ker A&&\subset\ker D^* \hspace{0.24\linewidth} &&B^*\ker A^*&&\subset\ker D \\
  &B\ker D&&\subset\ker A^* \hspace{0.24\linewidth} &&C^*\ker D^*&&\subset\ker A \\
  &A\ker C&&\subset\ker B^* \hspace{0.24\linewidth} &&A^*\ker B^*&&\subset\ker C \\
  &D\ker B&&\subset\ker C^* \hspace{0.24\linewidth} &&D^*\ker C^*&&\subset\ker B.
\end{alignat*}

Since $C$ is isometric on $\ker A$, this embedding is isometric, and there is a
complementary inclusion making $\ker A$ and $\ker D^*$ unitarily isomorphic.
Extending this idea to the other inclusions it follows that
\begin{align}\label{eq:dimkers}
  &\DK A^*=\DK D    &&\DK D^*=\DK A \\
  &\DK B^*=\DK C    &&\DK C^*=\DK B. \label{eq:dimkerss}
\end{align}
Note that this implies $\fred{A}=-\fred{D}$ and $\fred{C}=-\fred{B}$.

As in the proof of \Qref{lem:sixFeu} we consider the $+1$-eigenspace of $W$ as
the kernel of $F-\ch F\ch$, which is the off-diagonal part of $F$. A vector
writing in chiral components as $(\psi_+,\psi_-)$ lies in this kernel iff
$\psi_-\in\ker B$ and $\psi_+\in \ker C$. Hence
$\six_+(W)=\DK C -\DK B=\fred{C}$. Similarly, for the $-1$-eigenspace we need
the kernel of the diagonal part of $F$, getting $\six_-(W)=\DK A-\DK
D=\fred{A}$. The computations for $W'$ are analogous.
\end{proof}

\subsection{The Fredholm index of F}\label{sec:FredF}

Clearly, if the half-step operator $F$ is unitary its Fredholm index vanishes. If $F$ is merely essentially unitary we can express its Fredholm index in terms of the Fredholm indices of its chiral blocks:
\begin{lem}\label{lem:iF}
	Let $F$ be an essentially unitary half-step operator. Then its Fredholm index is given by
	\begin{equation}\label{eq:Findex}
		\fred{F}=\fred{B}+\fred{C}=\fred{A}+\fred{D}.
	\end{equation}
\end{lem}

\begin{proof}
  By essential unitarity of $F$, $A^*B+C^*D$ is compact. Moreover, since the half-step operator constitutes essentially gapped walks $W$ and $W'$, the chiral blocks are Fredholm operators, i.e.\ for each $X\in\{A,B,C,D\}$, there is an $\quasinv X$, such that $X\quasinv X-\idty$ and $\quasinv X X-\idty$ are compact. Hence, $D=-\quasinv{C^*}A^*B+K'=-C\quasinv C\quasinv{C^*}A^*B+K$, for some compact operators $K$ and $K'$.
  This allows the following factorisation:
	\begin{equation}\label{eq:factor}
		F=\begin{pmatrix}
		A & B\\ C & -C\quasinv C\quasinv{C^*}A^*B+K
		\end{pmatrix}=\begin{pmatrix}
		\idty & 0\\ 0 & C
		\end{pmatrix}\begin{pmatrix}
		A & \idty \\ \idty & -\quasinv C\quasinv{C^*}A^*
		\end{pmatrix}\begin{pmatrix}
		\idty & 0\\ 0 & B
		\end{pmatrix}+\begin{pmatrix}
		0 & 0\\0& K
		\end{pmatrix},
	\end{equation}
    All three factors are Fredholm operators, thus the Fredholm index of this expression is given by the sum of their Fredholm indices. The first and the third factor have the same indices as $C$ and $B$, respectively. Moreover, the Fredholm index of the middle factor vanishes. Indeed, let $M$ be any positive operator, then, for every $A$
	\begin{align}
	\ker\begin{pmatrix}
		A & \idty \\ \idty & -MA^*
	\end{pmatrix} &=\left\{(\varphi,\psi)\vert\,
	A\varphi+\psi=\varphi-MA^*\psi=0\right\}\\
	&= \left\{(MA^*\psi,\psi)\vert\, (\idty+AMA^*)\psi=0\right\}\\
	&= \left\{0\right\},
	\end{align}
	where the last equality follows from strict positivity of $\idty+AMA^*$.
A similar reasoning shows that also the kernel of the adjoint is trivial. Since $\quasinv C\quasinv{C^*}$ is positive, this guarantees, that the middle factor in \eqref{eq:factor} has trivial Fredholm index.
Thus,
    \begin{equation}
      \fred F=\fred{F+K}=\fred B+\fred C,
    \end{equation}
and the second equation in \eqref{eq:Findex} follows from \eqref{eq:six}.
\end{proof}

\subsection{Existence of $F$}\label{sec:existence}

A question of primary interest is whether for every chiral symmetric unitary $W$ there exists a half-step operator $F$ such that $W=\ch F^*\ch F$. In this paragraph we answer this question in the affirmative and, moreover, identify conditions which guarantee that a pair of chiral symmetric unitaries is a pair of timeframes for the same half-step operator $F$. Before diving into these existential questions, we examine the uniqueness of a given $F$ when only one of the timeframes is kept fixed.

\begin{lem}\label{lem:Forbit}
	Let $F_0$ be a half-step operator for $W$, i.e. $W=\gamma F_0^*\gamma F_0$. Then any other $F$ is also a half-step operator for $W$ if and only if $F=UF_0$, with $\gamma U=U\gamma$.
    In the chiral eigenbasis this means
	\begin{equation}
		F=\begin{pmatrix}
			U_+&0\\0&U_-
		\end{pmatrix}F_0.
	\end{equation}
    Similarly, the second timeframe $W'=F_0\gamma F_0^*\gamma$ is invariant if and only if $F=F_0U$ with $\gamma U=U\gamma$.
\end{lem}
\begin{proof}
	From $\gamma F^*\gamma F=\gamma F_0^*\gamma F_0$ we get
	\begin{equation}
		F^*\gamma F=F_0^*\gamma F_0\qquad\Leftrightarrow\qquad 	\gamma FF_0^*=FF_0^*\gamma.
	\end{equation}
	Hence, $F=(FF_0^*)F_0=UF_0$ with $\gamma U=U\gamma$. The statement for $W'$ follows analogously.
\end{proof}

In the following we write $W$ in the chiral eigenbasis as
\begin{equation}
W=\begin{pmatrix}
\alpha &\beta\\-\beta^*&\delta
\end{pmatrix},
\end{equation}
where the relation between the off-diagonal blocks is due to the chiral symmetry, which also forces $\alpha$ and $\delta$ to be self-adjoint \cite[Lemma 3.7]{UsOnTI}. Moreover, we denote by $P_\pm$ be the projections onto the $\pm1$-eigenspaces of $W$ and write $\HH$ as
\begin{equation}\label{eq:hh_decomp}
    \HH=P_-\HH\oplus P_+\HH\oplus\KK,
\end{equation}
where $\KK=(P_++P_-)^\perp \HH$ is the subspace on which $W$ is properly gapped, and by $\Gamma_\pm$ the projections onto the $\pm1$-eigenspaces of $\gamma$, i.e. $\gamma=\Gamma_+-\Gamma_-$.

\begin{prop}\label{prop:Fexists}
	Let $W$ be a chiral symmetric essentially gapped unitary. Then there exists a half-step operator $F$, such that
	\begin{equation}\label{eq:F2W}
	W=\gamma F^*\gamma F.
	\end{equation}
\end{prop}
\begin{proof}
	We prove the existence of $F$ by explicitly constructing it on each of the subspaces in \eqref{eq:hh_decomp}
	The finite dimensional $\pm1$-eigenspaces of $W$ can then further be split into
	\begin{align}\label{eq:eigenspaces}
	&P_-\HH=\ker(\alpha+1)\oplus\ker(\delta+1)=\vcentcolon\ker A\oplus\ker D\\ &P_+\HH=\ker(\alpha-1)\oplus\ker(\delta-1)=\vcentcolon\ker C\oplus\ker B
	\end{align}
	where the direct sums refer to a splitting with respect to $(\Gamma_+\oplus\Gamma_-)\HH$ and we wrote $\ker(\alpha\pm1)$ instead of $\ker(\alpha\pm\Gamma_+)$ and $\ker(\delta\pm1)$ instead of $\ker(\delta\pm\Gamma_-)$ in order to streamline notation. The second equalities identify the obtained spaces with the kernels of the matrix blocks of an $F$, which is to be constructed. They follow from arguments similar to those in the proof of \Qref{lem:sixFeu}. Since we did not fix a second timeframe, which would fix the co-kernels, i.e.\ the kernels of $A^*,B^*,C^*$ and $D^*$, we are free to choose these finite dimensional subspaces. We do this by choosing some finite dimensional unitaries, such that
	\begin{align}
	&C^-\ker A=\vcentcolon\ker D^* && B^-\ker D=\vcentcolon\ker A^*\label{eq:kerAkerD}\\
	&A^+\ker C=\vcentcolon\ker B^* && D^+\ker B=\vcentcolon\ker C^*\label{eq:kerCkerB},
	\end{align}
	which is certainly possible due to the equality of the respective dimensions according to the proof of \Qref{lem:sixFeu} and the fact, that on both side, the kernels are pairwise orthogonal.	Denoting by $\KK'$ the complement of the co-kernels, we define a unitary $V$ via $V\KK=\KK'$, where we choose $V$ in a way, such that $V\left({\Gamma_\pm|}_{\KK}\right) V^*={\Gamma_\pm|}_{\KK'}$ which guarantees
	\begin{align}
	V\gamma_\KK V^*=\gamma_{\KK'}.
	\end{align}
	For the choice above note, that $\Gamma_\pm|_\KK$ and $\Gamma_\pm|_{\KK'}$ are well defined, since $P_\pm$ and $\Gamma_\pm$ commute pairwise.
	From the proof of \Qref{lem:sixFeu} we also know, that $F$ has to map the kernels of $A,B,C,D$ isometrically onto the corresponding co-kernels. Hence, we can proceed by constructing separate unitaries $F=F_-\oplus F_+\oplus F_\KK$ with
	\begin{align}
	F_-&\colon\ker(A)\oplus\ker(D)\to\ker(A^*)\oplus\ker(D^*)\\
	F_+&\colon\ker(C)\oplus\ker(B)\to\ker(B^*)\oplus\ker(C^*)\\
	F_\KK&\colon\KK\to\KK'.
	\end{align}
	
	For $F_\pm$ we can make use of the choices we already made: Setting
	\begin{align}\label{eq:Fminus}
	F_-=\begin{pmatrix}
	0&B^-\\C^-&0
	\end{pmatrix}\qquad\text{we get}\qquad
	\gamma F_-^*\gamma' F_-=-P_-
	\end{align}
	as needed. Note, that one has to keep track on which subspace $\gamma$ is given during the product, which we indicated by distinguishing between $\gamma=\operatorname{diag}(\idty_{\ker A},-\idty_{\ker D})$ and $\gamma'=\operatorname{diag}(\idty_{\ker A^*},-\idty_{\ker D^*})$, which are different for non-vanishing $\six_-(W)=\DK A-\DK A^*$.
	
	Similarly, with
	\begin{align}\label{eq:Fplus}
	F_+=\begin{pmatrix}
	A^+&0\\0&D^+
	\end{pmatrix}\qquad\text{we get}\qquad
	\gamma F_+^*\gamma'F_+=P_+
	\end{align}
	where in this case $\gamma$ and $\gamma'$ have the same form in both subspaces. Note, that we are completely free in the choice of the finite dimensional unitaries $A^+,D^+,B^-$ and $C^-$.
	
	This leaves us with the task of finding an appropriate $F_\KK$, which we simply denote by $F$ in the following construction. Instead of directly considering $F\colon\KK\to\KK'$, define $\widetilde F\colon\KK\to\KK$ via
	\begin{align}\label{eq:Ftwiddle}
	\widetilde F=\frac{1}{\sqrt{2}}\begin{pmatrix}
	\sqrt{1+\alpha}&V_\beta\sqrt{1-\delta}\\
	-V_\beta^*\sqrt{1-\alpha}&\sqrt{1+\delta}.
	\end{pmatrix},
	\end{align}
	which is well defined, because on the complement of $\ker(\alpha\pm1)$ and $\ker(\delta\pm1)$, $1\pm\alpha$ and $1\pm\delta$ are positive and $\beta$ is invertible, providing a unique polar isometry $V_\beta$. The unitarity of $F$ then follows from the unitarity conditions of $W$.
	Using $\sqrt{\Gamma_+-\alpha}V_\beta=V_\beta\sqrt{\Gamma_--\delta}$ and $V_\beta\sqrt{\Gamma_--\delta^2}=\sqrt{\Gamma_+-\alpha^2}V_\beta=\beta$, which also follow from the unitarity of $W$, we find, that $\widetilde F$ is a square root for $W$ on $\KK$, i.e. $\widetilde F^2=W_\KK$. Moreover, abbreviating $\gamma_\KK=\gamma$ and $\gamma_{\KK'}=\gamma'$, we get $\gamma\widetilde F\gamma=\widetilde F^*$, i.e. $\widetilde F$ itself is chiral symmetric, which implies $\gamma \widetilde F^*\gamma\widetilde F=W_\KK$.
	Setting $F=V\widetilde F$, with $V$ from above, indeed gives
	\begin{align}
	\gamma F^*\gamma'F&=\gamma\widetilde F^* (V^*\gamma'V)\widetilde F=\gamma \widetilde F^*\gamma \widetilde F=W_\KK.
	\end{align}	
	Combining the three operators, we get $F=F_-\oplus F_+\oplus F_\KK\colon\HH\to\HH$, with $W=\gamma F^* \gamma F$.
\end{proof}
The block structure of  $F=F_-\oplus F_+\oplus F_\KK$ with respect to the eigenspaces of $W$ is not just a special example, which serves the proof of the proposition, but is actually valid for every given $F$. This \textbf{standard form} will be established in general later in \Qref{sec:completeness_proof}.
Note, that the existence of an $F$ for every $W$ hinges on the infinite dimension of the Hilbert space. For finite systems an additional condition, namely $\six_-(W)=0$ is necessary for the existence of $F$.

Having established the existence of a half-step operator $F$ for any chiral symmetric unitary $W$, we now focus on the possible second time frames. By \Qref{lem:Forbit} we know, that $W$ fixes $F$ only up to multiplication with a $\gamma$-commuting unitary from the left. However, this typically changes the second timeframe which raises the question whether a given second chiral symmetric unitary $W'$ lies in the orbit of $UF_0\gamma F_0^*\gamma U^*$. In other words: Which conditions do $W$ and $W'$ need to fulfil, in order to be considered as time frames of each other?

\begin{thm}\label{thm:twotimeframes}
	Let $W$ and $W'$ be two chiral symmetric essentially gapped unitaries for the same chiral symmetry $\ch$. Then the following are equivalent:
	\begin{itemize}
		\item[(1)] There is a half-step operator $F$, such that
		\begin{equation}\label{timeframeequations}
		W=\gamma F^*\gamma F\qquad\text{and}\qquad W'=F\gamma F^*\gamma.
		\end{equation}
		\item[(2)] $W$ and $W'$ fulfil
		\begin{equation}\label{eq:indexcondition}
		\six_\pm(W)=\pm\six_\pm(W')
		\end{equation}
		and there exists a unitary $U$ with
		\begin{equation}\label{eq:specialU}
		W'=UWU^*\qquad\text{and}\qquad U\gamma_\KK U^*=\gamma_{\KK'},
		\end{equation}
		where $\gamma_\KK$ and $\gamma_{\KK'}$ denote the chiral symmetry restricted to $\KK$ and $\KK'$ in \eqref{eq:hh_decomp}. I.e., in chiral eigenbasis, $U$ is a block diagonal mapping between the parts of the Hilbert space, where $W$ and $W'$ are gapped.
	\end{itemize}
\end{thm}

\begin{proof}
	(2)$\Rightarrow$(1): This proof direction will be quite similar to the proof of \Qref{prop:Fexists}, with the difference, that the second time frame fixes the kernels of $A^*,B^*,C^*$ and $D^*$. The identifications are now given by
	\begin{align*}
	&\ker A\phantom{*}=\ker(\alpha+1), &&\ker B\phantom{*}=\ker(\delta-1), &&&\ker C\phantom{*}=\ker(\alpha-1), &&&&\ker D\phantom{*}=\ker(\delta+1),\\
	&\ker A^*=\ker(\alpha'+1), &&\ker B^*=\ker(\alpha'-1), &&&\ker C^*=\ker(\delta'-1), &&&&\ker D^*=\ker(\delta'+1).
	\end{align*}
	Thereby, the necessary dimension-equalities in \eqref{eq:dimkers} are guaranteed by the unitary equivalence and the index condition \eqref{eq:indexcondition}. Indeed, the unitary equivalence of $W$ and $W'$ guarantees $\rank P_\pm=\rank P'_\pm$, i.e.
	\begin{align}\label{eq:ranks}
	\dk(\alpha\pm1)+\dk(\delta\pm1)=\dk(\alpha'\pm1)+\dk(\delta'\pm1).
	\end{align}
	Moreover, using \eqref{eq:eigenspaces}, we can express the symmetry indices of $W$ as $\six_\pm(W)=\dk(\alpha\mp 1)-\dk(\delta\mp 1)$	and similarly for $W'$. The index condition \eqref{eq:indexcondition} then reads
	\begin{align}\label{eq:dimkerindexconditions}
	\dk(\alpha\mp 1)-\dk(\delta\mp 1)=\pm\bigl(\dk(\alpha'\mp 1)-\dk(\delta'\mp 1)\bigr).
	\end{align}
	Combining \eqref{eq:dimkerindexconditions} with \eqref{eq:ranks}, we conclude
	\begin{align}\label{eq:Wdimkers}
	\dk(\alpha+1)&=\dk(\delta'+1) && \dk(\alpha-1)=\dk(\alpha'-1)\\
	\dk(\delta+1)&=\dk(\alpha'+1) && \dk(\delta-1)=\dk(\delta'-1),\nonumber
	\end{align}
	in accordance with \eqref{eq:dimkers}, together with the identifications displayed above.
	
	We can now use the construction of $F_\pm$ from above, noting, that \eqref{eq:Fminus} and \eqref{eq:Fplus} yield $-P_-'$ and $P_+'$, when we evaluate $F_\pm\gamma F_\pm^*\gamma'$ instead of $\gamma F_\pm^*\gamma' F_\pm$.
	For $F_\KK$ we again use the same ansatz via $\widetilde F$. By assumption, $U$, restricted to $\KK\to\KK'$, fulfils $U^*\gamma' U=\gamma$. Hence, setting $F=U\widetilde F$ gives $\gamma F^*\gamma'F=W_\KK$ as before. Moreover, since $\widetilde F$ is a chiral symmetric square root of $W_\KK$, i.e.\ $W_\KK=\gamma\widetilde F^*\gamma\widetilde F=\widetilde F\gamma\widetilde F^*\gamma=\widetilde F^2$, we also get
	\begin{equation}
		F\gamma F^*\gamma'=U \widetilde F\gamma\widetilde F^*(U^*\gamma' U)U^*=U\widetilde F\gamma\widetilde F^*\gamma U^*=UW_\KK U^*=W'_{\KK'}.
	\end{equation}

	(1)$\Rightarrow$(2): The index relation directly follows from \eqref{eq:six}, \eqref{eq:sixprime}	and \eqref{pmindices}. For the unitary equivalence, we need to show, that the unitary equivalence of $W$ and $W'$ via $F$ also guarantees the existence of a unitary $U$, with the extra condition in \eqref{eq:specialU}. Between the $\pm$-eigenspaces we can just take $U=F$. On their complement, however, we need to fulfil $U\gamma U^*=\gamma'$, which $F$ certainly doesn't.
	So let us restrict the considerations to $\KK$ and $\KK'$. In order to streamline the notation a bit, we drop the $\KK$ and $\KK'$ suffixes and indicate the space we are currently in by writing e.g.\ $\gamma'$ instead of $\gamma$.
	On $\KK$ we have $W=\gamma F^*\gamma'F=\gamma\widetilde F^*\gamma\widetilde F$, with $\widetilde F$ from the construction above. From this we get $F\widetilde F^*\gamma=\gamma'F\widetilde F^*$, which in turn implies
	\begin{align}
	   W'  &=F\gamma F^*\gamma'=(F\widetilde F^*)\widetilde F\gamma\widetilde F^*(F\widetilde F^*)^*\gamma'\\
	       &=(F\widetilde F^*)\widetilde F\gamma\widetilde F^*\gamma(F\widetilde F^*)^*\\
	       &=(F\widetilde F^*)W(F\widetilde F^*)^*,
	\end{align}
	where in the last step we again used that $\widetilde F$ is a chiral symmetric square root for $W_\KK$.
	Hence, combining $U=F(P_++P_-+\widetilde F^*)$, we get the unitary equivalence $W'=UWU^*$, with a unitary intertwining $\gamma_\KK$ and $\gamma_{\KK'}$, as needed.
\end{proof}
Note, that the condition on $U$ in \eqref{eq:specialU} might be an artefact of the proof strategy. In fact, it is always met in a timeframed setting and we could not find a counterexample of two unitarily equivalent walks, which are chiral symmetric for the same symmetry, where \eqref{eq:specialU} is violated. Hence, the extra condition might well be redundant.

\section{Introducing locality} \label{sec:locality}

So far we considered half-step operators $F$ and the corresponding $W$ and $W'$ in \eqref{eq:timeframes} without taking the spatial structure of the underlying Hilbert space into account. The above results are thus valid on arbitrary separable Hilbert spaces. However, our goal is to classify physical systems obeying a locality condition on Hilbert spaces of the form \eqref{eq:H}.
Our standing assumption will be that $F$, and consequently the timeframe unitaries $W,W'$ in \eqref{eq:timeframes} are \textbf{essentially local}. This means that for some (hence for all) $a\in\Ir$ the commutator $[F,\Pg a]$ with the half-space projection $\Pg a$ is a compact operator. We denote by $\AAloc\subset\BB(\HH)$ the set of operators satisfying this condition. It is easy to see that $\AAloc$ is a norm closed operator algebra, which makes some of the constructions below very easy. Let us collect some of the salient features.

\begin{prop}\label{prop:inx0}
Let $F\in\AAloc$ be unitary, and denote by $P=\Pg a$ some half-space projection. Then
\begin{itemize}
\item[(1)] $PFP$ is an essentially unitary operator on $P\HH$, and thus has a Fredholm index, which we denote by
\begin{equation}
  \fredP F:=\fred{PFP}=\dim\ker{PFP}-\dim\ker{PF^*P}\in\Ir.
\end{equation}
\item[(2)] $\fredP\cdot$ is norm continuous on the unitary group of $\AAloc$, and satisfies the product rule\\ $\fredP{F_1F_2}=\fredP{F_1}+\fredP{F_2}$.
\item[(3)] The following are equivalent:
  \begin{itemize}
  \item[(a)] $\fredP F=0$
  \item[(b)] There is a norm continuous unitary path $[0,1]\ni t\mapsto F(t)\in\AAloc$ such that $F(0)=\idty$ and $F(1)=F$.
  \item[(c)] $F$ results from a continuous driving, i.e., it is connected to the identity by a path satisfying the differential equation
             $\partial_t F(t)=i H(t)F(t)$, where $H(t)=H(t)^*\in\AAloc$, $\norm{H(t)}$ is bounded and $t\mapsto H(t)$ is measurable or, alternatively, piecewise constant.
  \end{itemize}
\end{itemize}
\end{prop}

The index $\fredP{\cdot}$ was first introduced in \cite{Kitaev} for banded unitaries as a kind of net information flow across the point $0$. In that case and, more generally, when $[F,P]$ is even a Hilbert-Schmidt operator, it can be expressed by a simple formula, which, however, obscures somewhat that it evaluates to an integer. When $P_x$ is the projection onto the cell $\HH_x$, it is
\begin{equation}\label{fredFK}
  \fredP F=\sum_{x,y:\ x<0\leq y} \tr(P_xF^*P_yF)-\tr(P_yF^*P_xF).
\end{equation}
In any case, this is well-defined and finite for band matrices. Similarly, if a Hamiltonian $H$ is a banded matrix, with a maximal length $L$ such that $P_xHP_y=0$ for$\abs{x-y}>L$, it is also essentially local, and thus produces unitary evolutions $F=\exp(itH)\in\AAloc$ with vanishing index. Since the exponential function is not a finite degree polynomial, $F$ will not be a banded matrix itself. This is why essential locality is a more manageable condition than bandedness, especially for the construction of suitable continuous paths. In this spirit, the key statement in (3) is that vanishing index indeed completely captures the idea of the existence of an essentially local driving.

\begin{proof}[Sketch of proof of \Qref{prop:inx0}]
The main ideas have been explained elsewhere in increasing generality in \cite{old_index_paper,LongVersion}, so we mostly just sketch the new features.

(1),(2): $PF^*P$ is an essential inverse for $PFP$, because $P-PF^*PPFP=PF^*[F,P]P$ is compact. Continuity and the product formula thus follow from general properties of the Fredholm index.

(3c)$\Leftrightarrow$(3b): The Picard-Lindel\"of iteration for the differential equation lives entirely in the set of norm continuous functions $[0,1]\to\AAloc$. The necessary Lipshitz condition follows from the boundedness of $\norm{H(t)}$, so we can conclude the existence of a continuous solution $F_t$. Conversely, suppose that there is a continuous path. We can then find intermediate points $0=t_0<t_1<\cdots<t_r=1$ so that $\norm{F(t_{k+1})-F(t_k)}<2$. We claim that on each such interval we can find a constant effective Hamiltonian $H_k$ so that $F(t_{k+1})=\exp\bigl(i(t_{k+1}-t_k)H_k \bigr)F(t_k)$. By multiplying from the right with $F(t_k)^*$ this reduces to the general statement that in any C*-algebra $\AA$, a unitary $U\in\AA$ with $\norm{U-\idty}<2$ can be written as $\exp(iH)$ with $H\in\AA$ and $\norm H<\pi$.
This is, however, just an application of the functional calculus. When $\norm{U-\idty}<2$, $-1$ is not in the spectrum of $U$, so on the spectrum of $U$ the logarithm function with a branch cut on the negative real axis is continuous, so we can set $H=-i\log U\in\AA$.

(3b)$\Leftrightarrow$(3a):  The trivial direction follows by continuity and $\fredP\idty=0$. For the converse, let $\fredP F=0$. We then construct a path in two stages: First we connect $F$ to a unitary $F'$, for which exactly commutes with $P$, i.e., $F'=F_L\oplus F_R$ with unitaries $F_R$ on $P\HH$ and $F_L$ on $(\idty-P)\HH$. We then separately connect these unitaries to the respective identities. This automatically preserves (exact) essential locality, and is possible by the measurable functional calculus, i.e., by keeping the spectral family of the operator fixed and merely deforming all eigenvalues continuously to $1$. No spectral gap or branch cut is needed here, and a possible eigenvalue at $-1$ can be deformed to $+1$ along wither the top or the bottom half circle. The norm continuity of the path only requires that the deformation of each spectral value is continuous. It remains to construct a continuous decoupling of $F$ to $F'$. This is discussed at great length in \cite[Theorem VII.4]{LongVersion} and will not be reproduced here.
\end{proof}

Hence, by \Qref{prop:inx0} $(3)$, the index $\fredP F$ distinguishes the settings $(H)$ and $(F)$ in \Qref{sec:walks_protocols}. It has sometimes been argued that the existence of a continuous driving is a general requirement of physics: How else is such a process to be realized? The answer is that often the physical implementation uses additional degrees of freedom, i.e., the implementation uses a larger Hilbert space. For example, it is impossible to implement a shift (lattice translation by one step) by local operators. Invariably, any norm continuous connection of a shift with the identity has to violate essential locality. But with a second copy of the system this can be realized by local swap operations between the two copies, which can be contracted to the identity. This is exactly the effect measured by the above index \cite{old_index_paper}. 

An important consequence of essential locality of the half-step operator $F$ is the essential locality of its chiral blocks $A,B,C,D$: It implies that their right Fredholm indices are well-defined. Moreover, since the half-chain projection $PFP$ of an essentially local $F$ is essentially unitary on the half-chain the above results which are valid in the merely essentially unitary setting yield statements about systems confined to a half-chain after replacing the Fredholm indices by right Fredholm indices.

In this way, \Qref{lem:sixFeu} on the one hand expresses the symmetry indices $\six(W)$ and $\six(W')$ of the walks in \eqref{eq:timeframes} in terms of the Fredholm indices $\fred A,\fred B,\fred C$ and $\fred D$. On the other, by applying it to the essentially unitary half-chain walks $PWP$ and $PW'P$ it also determines the right symmetry indices $\sixR(W)$ and $\sixR(W')$ in terms of the right Fredholm indices $\fredP A,\fredP B,\fredP C$ and $\fredP D$, i.e.
\begin{align}
\sixR(W) &=\fredP{A}-\fredP{B}=\fredP{C}-\fredP{D}\label{eq:sixR}\\
\sixR(W') &=\fredP{C}-\fredP{A}=\fredP{D}-\fredP{B}\label{eq:sixRprime}.
\end{align}

Similarly, \Qref{lem:iF} expresses the right Fredholm index of $F$ in terms of $\fredP A,\fredP B,\fredP C$ and $\fredP D$. This result is closely related to \cite{Asbo3}, in which \cite[Eq.(14)]{Asbo3} is just \eqref{eq:Findex} applied to half-chain walks with the additional assumption $\fredP{F}=0$, and the invariants $\nu_0$ and $\nu_\pi$ defined in \cite[Eq.(17)]{Asbo3} correspond to $-\fredP{B}$ and $-\fredP{D}$. Unfortunately, the assumption $\fredP{F}=0$ does not hold for typical examples of halfstep operators of quantum walks (e.g. for the split-step walk, see \eqref{eq:splitstep-F}), so that a direct application of the resulting index formulas in \cite{obuse2015unveiling} is, at first sight, problematic. However, it can be justified in our framework for settings in which two bulks are joined (see \Qref{sec:time_framed_walks}).

In contrast to \Qref{lem:iF}, \Qref{lem:sixFu} requires exact unitarity of $F$, $W$ and $W'$, and has therefore nothing to say about the corresponding half-chain operators. This is in accordance with the observation that the symmetry indices like $\sixR_\pm$ are neither invariant under compact nor under gentle perturbations \cite{LongVersion}. We collect the connections between the symmetry indices of the quantum walks $W$ and $W'$ and the Fredholm indices $\fred\cdot$ and $\fredP\cdot$ of the matrix blocks of $F$ in \Qref{sec:results}.

Let us finally address the locality of the half-step operator $F$ constructed from a given chiral symmetric $W$ in the proof of \Qref{prop:Fexists}.

\begin{schol}\label{schol:esslocF}
	Let $W$ be a chiral symmetric walk. Then the half-step operator constructed in the proof of \Qref{prop:Fexists} can be chosen essentially local.
\end{schol}

To show this, we only need to consider $F_\KK$, since $F_+$ and $F_-$ have finite rank and can therefore be neglected. First note, that the $\widetilde F$ in \eqref{eq:Ftwiddle} is given entirely in terms continuous functions of matrix blocks of $W$ or their polar isometries. Therefore it is essentially local by the following lemma:

\begin{lem}\label{lem:esslocpolarisom}
	Let $X\in\AAloc$ be a Fredholm operator. Then its polar isometry $U_X$, as well as its absolute value $|X|=\sqrt{X^*X}$ are essentially local.
\end{lem}
\begin{proof}
	On the one hand $\AAloc$ is norm-closed, so the essential locality of $|X|$ follows via the Weierstraß-theorem.
 	On the other, let $P$ be a half-chain projection on $\HH$. Then, by $[X,P]=[U_X|X|,P]=U_X[|X|,P]+[U_X,P]|X|$, $[U_X,P]|X|$ is the difference of compact operators. Moreover, since $X$ is Fredholm, the only way for $[U_X,P]|X|$ to be compact is that $[U_X,P]$ is compact, i.e $U_X$ is essentially local.
\end{proof}

It remains to show, that the unitary $V\colon\KK\to\KK'$ in the proof of \Qref{prop:Fexists} can be chosen essentially local, which boils down to the following lemma:

\begin{lem}\label{lem:esslocisometry}
	Let $\HH$ be a Hilbert space with a one-dimensional lattice structure, as in \eqref{eq:H}. Moreover, let $Q_1$ and $Q_2$ be two finite rank projections on $\HH$. Then there exists an essentially local partial isometry $\Lambda$, with
	\begin{align}\label{eq:partialisom}
	\Lambda^*\Lambda=\idty-Q_1\qquad\text{and}\qquad \Lambda\Lambda^*=\idty-Q_2.
	\end{align}
\end{lem}

In order to see how this applies to $V\colon\KK\to\KK'$, let $\HH_\pm=\Gamma_\pm\HH$ and similarly for $\KK_\pm$ and $\KK'_\pm$. Then, since $\KK_\pm$ and $\KK'_\pm$ are the complements of finite dimensional subspaces of $\HH_\pm$, by \Qref{lem:esslocisometry} we find essentially local isometries $V_+$ and $V_-$, such that
\begin{equation}
V_\pm^* V_\pm\HH_\pm=\KK_\pm\qquad\text{and}\qquad V_\pm V_\pm^*\HH_\pm=\KK'_\pm.
\end{equation}
Considering $V_\pm$ as mappings between $\KK_\pm$ and $\KK_\pm'$ they become unitary and define the desired essentially local unitary $V$ via
\begin{equation}
V=\begin{pmatrix}
V_+&0\\0&V_-
\end{pmatrix}.
\end{equation}
\begin{proof}[Proof of \Qref{lem:esslocisometry}]
	The proof can be reduced to the construction of an essentially local partial isometry $\Pi$, with $\Pi^*\Pi=\idty-N$ and $\Pi\,\Pi^*=\idty$, for an arbitrary finite rank projection $N$. Then, with $\Pi_1$ and $\Pi_2$ being such isometries for $Q_1$ and $Q_2$, respectively, the statement follows from setting $\Lambda=\Pi_2^*\Pi_1$.
	
	To construct $\Pi$, let $n=\rank N$ and
	\begin{equation}
	\Pi_M=(\idty-P)+ P{S^*}^nP,
	\end{equation}
	where $P=P_{>0}$ is the projection onto $\bigoplus_{x>0}\HH_x$ and $S=\sum_{x,i}\ketbra{x+\delta_{1i},i}{x,i}$ denotes the partial bilateral shift which shifts the first basis vector in each cell $\HH_x=\operatorname{span}\{\ket{x,i}:i=1,\dots,d_x\}$ to the right. Sandwiching $S$ with $P$, we get the unilateral shift, with
	\begin{equation}
	(PS^*P)(PSP)=P\qquad\text{and}\qquad(PSP)(PS^*P)=P-\ketbra{1,1}{1,1}.
	\end{equation}
	Extending these properties to $S^n$, we get
	\begin{align}
	\Pi_M\Pi_M^*&=(\idty-P)+P=\idty\\
	\intertext{and}
	\Pi_M^*\Pi_M&=(\idty-P)+(PS^nP)(P{S^*}^nP)=(\idty-P)+(P-M)=\idty-M,
	\end{align}
	with $M=\sum_{k=1}^n\ketbra{k,1}{k,1}$.
	$\Pi_M$ is defined as the identity on one half of the chain and as a unilateral shift on the other, wherefore it is automatically essentially local.
	
	Since $\rank M=n=\rank N$, there exists a unitary $U$ which identifies the subspaces $N\HH$ and $M\HH$, i.e.\ $N=U^*MU$, and acts like the identity on their complement. Being a finite rank perturbation of the identity, $U$ is clearly essentially local. Hence, $\Pi=\Pi_MU$ is essentially local with
	\begin{equation}
	\Pi^*\Pi=U^*\Pi_M^*\Pi_MU=U^*(\idty-M)U=\idty-N\qquad\text{and}\qquad \Pi\,\,\Pi^*=\Pi_M\Pi_M^*=\idty.
	\end{equation}
\end{proof}

\section{Proof of Theorem~\ref{thm:complete}}\label{sec:complete_indices}
In the previous sections we defined ten indices for half-step walks $F$: On the one hand, there are the Fredholm indices $\fred{\cdot}$ of $F$ and its four chiral blocks $A,B,C$ and $D$.
On the other, there are the right Fredholm indices $\fredP{\cdot}$ of these five operators. However, we also saw that these ten indices are not independent of each other. In the following we identify a subset of five indices and prove its invariance under norm continuous, as well as compact perturbations. Moreover, we define a generating example, with which every index combination in $\Ir^5$ can be realized. On the one hand, this implies the independence of these five indices and on the other, that there are no empty classes. This proves the first part of \Qref{thm:complete}.

Afterwards, we prove the completeness result in \Qref{thm:complete} in several steps: First we introduce a standard form for half-step walks $F$ which facilitates the proof. Then, we construct a homotopy which connects each half-step walk with one whose timeframed walks have essential spectrum only at $\pm i$. Finally, we assemble the proof of completeness in the last part of this section.

\subsection{Independent set of indices}\label{sec:independent_set}
The results of the previous sections imply the following dependencies between the ten indices of $F$: By \eqref{eq:six} or \eqref{eq:sixprime} in \Qref{lem:sixFeu} we can drop the index of one of the matrix blocks, w.l.o.g.\ $\fred{D}$. Since the lemma is formulated for merely essentially unitary half-step operators, we can also drop $\fredP{D}$. Furthermore, by \eqref{pmindices} in \Qref{lem:sixFu}, we can drop $\fred{C}$. This lemma requires exact unitarity and hence has nothing to say about right Fredholm indices. Finally, \Qref{lem:iF}, which is again valid also for merely essentially unitary half-step operators, allows us to drop $\fred{F}$ and $\fredP{F}$. Thus, the following set of integer-valued indices is independent:
  \begin{equation}\label{eq:indexset}
    \bigl(\fred{A},\fred{B},\fredP{A},\fredP{B},\fredP{C}\bigr).
  \end{equation}
	Moreover, these indices are invariant under continuous deformations of $F$, which do not break the assumptions from \Qref{sec:walks_protocols}, as well as compact perturbations. In the following, we take this as the standard set of independent indices for $F$.

	Any homotopy of $F$ must also be a homotopy of $A,B,C$ and $D$ and keeping
	the essential gap condition for $W$ and $W'$ is equivalent to $A,B,C$ and
	$D$ being Fredholm operators along any allowed path. Hence, since $\fred{X}$
	and $\fredP{X}$ for $X\in\{A,B,C,D\}$ are defined as Fredholm
	indices, they are constant along allowed paths.	The same holds for compact perturbations.

The independence of the five indices in \eqref{eq:indexset} is proved by defining a generating example which allows us to realize every index combination in $\Ir^5$:

\subsection{The generating example}\label{sec:generating_example}

Let $S$ be the unilateral shift on $\ell^2(\Ir)$ with $\fred{S}=0$ and $\fredP{S}=-1$, and consider the unitary operators
\begin{equation}\label{eq:pre_generating_example}
U(n,m)=\frac1{\sqrt{2}}\begin{pmatrix}
S^n &   S^m \\
-S^{-m}    &  S^{-n}
\end{pmatrix}\qquad\text{and}\qquad T(k)=\begin{pmatrix} S^k & 0\\ 0 &
1\end{pmatrix},\qquad n,m,k\in\Ir.
\end{equation}
It is straightforward to see that $\fredP{{U(m,n)}}=0$ for all $m,n\in\Ir$. Thus, by \cite[Theorem VII.4]{LongVersion} $U$ can be decoupled gently, i.e. there exists a chiral symmetric homotopy $U(m,n)\mapsto U_L(n,m)\oplus U_R(n,m)$, where $U_L(n,m)$ and
$U_R(n,m)$ are exactly unitary quantum walks on the left and the right half-chain, respectively. We define the
\textbf{generating example} as
\begin{equation}\label{eq:generating_example}
F(n_L,m_L,n_R,m_R,k)=T(k)\bigl(U_L(n_L,m_L)\oplus U_R(n_R,m_R)\bigr).
\end{equation}
Using $\fred{S}=0, \fredP{S}=-1$ and $\fred{S^a(P^\perp S^bP^\perp\oplus
PS^cP)\bigr}=b-c$ we can determine the index tuple \eqref{eq:indexset} of $F$:
\begin{equation}
    \bigl(\fred{A},\fred{B},\fredP{A},\fredP{B},\fredP{C}\bigr)=\bigl(n_L-n_R,m_L-m_R,-n_R-k,-m_R-k,m_R\bigr),
\end{equation}
i.e. the indices of $F$ may be calculated by applying the integer matrix
\begin{equation}
M=\left(\begin{array}{ccccc} 1 & 0 & -1 & 0 & 0\\
0 & 1 & 0 & -1 & 0\\
0 & 0 & -1 & 0 & -1\\
0 & 0 & 0 & -1 & -1\\
0 & 0 & 0 & 1 & 0
\end{array}\right)
\end{equation}
to the parameter vector $(n_L,m_L,n_R,m_R,k)$. Clearly $\abs{\det M}=1$ and, moreover, its inverse is also an integer matrix. Hence, for any possible combination of indices in \eqref{eq:indexset} we can construct an $F$ via \eqref{eq:generating_example} with parameters obtained by applying $M^{-1}$ to this index set. The generating example will also be used later in \Qref{sec:time_framed_walks} for the proof of \Qref{cor:(si+si')/2}.

\subsection{Standard form}

The $F$ we constructed in the proofs of \Qref{prop:Fexists} and \Qref{thm:twotimeframes} had a special structure, based on the eigenspace decompositions of $W$ and $W'$. This structure was not just an artefact of the proof technique, but can actually be raised to a \textbf{standard form} for every given half-step operator $F$.  Let $\HH=\Gamma_+\HH\oplus\Gamma_-\HH=\HH_+\oplus\HH_-$ be the decomposition of $\HH$ into the $\gamma$-eigenspaces. Similar to the splitting in the proofs mentioned above these eigenspaces can be further decomposed in two different ways, namely
\begin{alignat}{2}\label{eq:Hstructure1}
	\HH_+&=\ker A\oplus\ker C\oplus\Gamma_+\KK &&=\ker A^*\oplus\ker
	B^*\oplus\Gamma_+\KK'\\
	\HH_- &= \ker  D\oplus\ker B\oplus\Gamma_-\KK &&=\ker D^*\oplus\ker
	C^*\oplus\Gamma_-\KK',\label{eq:Hstructure11}
\end{alignat}
where by \Qref{lem:sixFeu} $\ker X$ is finite dimensional for $X=A,B,C,D$ due to the essential gap condition of the corresponding walks. The orthogonality of the direct summands is guaranteed by unitarity of $F$ and $F^*$ in \eqref{eq:unitarity_of_F}. For example, $\ker A$ and $\ker C$ are orthogonal by $A^*A+C^*C=\idty$.
Reordering the direct summands in \eqref{eq:Hstructure1} and \eqref{eq:Hstructure11} and using the arguments in the proof of \Qref{lem:sixFu}, $F$ takes the following form:
\begin{equation}\label{eq:Hstructure2}
	F\colon  (\ker A\oplus\ker
	D)\oplus(\ker C\oplus\ker B)\oplus\KK
	\to (\ker A^*\oplus\ker
	D^*)\oplus(\ker B^*\oplus\ker C^*)\oplus\KK'
\end{equation}
\begin{equation}\label{eq:Fstructure}
	F=
	\begin{pmatrix}0 & B^-\\ C^- & 0\end{pmatrix}
	\oplus
	\begin{pmatrix}A^+ & 0\\ 0 & D^+\end{pmatrix}
	\oplus
	\begin{pmatrix}A_\KK & B_\KK\\C_\KK & D_\KK\end{pmatrix}.
\end{equation}
Here, $A^+\colon\ker C\to\ker B^*,B^-\colon\ker D\to\ker
A^*,C^-\colon\ker
A\to\ker D^*$ and $D^+\colon\ker B\to\ker C^*$ are finite dimensional unitaries because they are isometries between spaces of the same dimension, see \eqref{eq:dimkers} and \eqref{eq:dimkerss}. $A_\KK,B_\KK,C_\KK$ and $D_\KK$ on the other hand are invertible operators on the remaining infinite dimensional Hilbert space. However, note that the finite dimensional unitaries $A^+,B^-,C^-$ and $D^+$ do not necessarily have to be of the same size. Indeed, according to \Qref{lem:sixFu} the difference in dimension of $A^+$ and $D^+$ constitutes the index $\six_+(W)=\six_+(W')$, whereas that of $B^-$ and $C^-$ determines $\six_-(W)=-\six_-(W')$. For exactly gapped walks $W$ and $W'$ the finite dimensional blocks vanish, and we get $\KK_+=\KK_+'=\HH_+$ and $\KK_-=\KK_-'=\HH_-$.

\subsection{The flattening construction}
For the classification of quantum walks in \cite{LongVersion} it turned out to be useful to flatten
their spectrum, i.e.\ to continuously deform it to an operator, whose spectrum is contained in $\{\pm 1,\pm i\}$
with only finitely degenerated eigenvalues at $\pm1$. We call such operators \textbf{essentially flatband}. In the current setting, flattening the spectrum of one of the timeframed walks automatically implies that the other timeframe is also essentially flatband by unitary invariance of the spectrum. The essential flatband condition has the following consequences for $F$:
\begin{lem}\label{lem:flatness}
	Let $F$ be the half-step walk with corresponding walks $W$ and $W'$. Then the following are equivalent:
	\begin{itemize}
		\item[(1)] $W$ and $W'$ are essentially flatband.
		\item[(2)] $\sqrt2X$ is essentially unitary for $X\in\{A,B,C,D\}$.
	\end{itemize}
	\begin{proof}
		It turns out easier to formulate the conditions in the Calkin algebra, as
		we get rid of the adjective ``essential'' in this way without loss of
		generality. In the Calkin algebra $W$ and $W'$ are exactly flat-band, which
		means, that $W=-W^*$, and the same for $W'$. But this is equivalent to
		$A^*A=C^*C,\,B^*B=D^*D,\, AA^*=BB^*$ and $DD^*=CC^*$ by \eqref{eq:walks}, which, by unitarity
		of $F$ in \eqref{eq:unitarity_of_F}, is equivalent to $\sqrt{2}X$ being unitary for $X\in\{A,B,C,D\}$.
	\end{proof}
\end{lem}

As for \Qref{lem:sixFeu} this lemma also holds when we drop all ``essentially''s. Similarly, it does not make use of the locality of the involved operators. However, for the above definition of essentially flatband operators exact unitarity is a crucial assumption as it allows one to speak of ``finitely degenerate eigenvalues'' at $\pm 1$. Relaxing this definition to mere essential unitarity does not make sense since essentially unitary operators are not normal in general. To be able to nevertheless speak of essentially gapped essentially unitary operators, we could instead take the finite dimensionality of the kernels of $W\pm \idty$ as a defining property. This would allow us to define essentially flatband essentially unitary operators as essentially gapped essentially unitary operators whose essential spectrum consists of $\pm i$.

However, a much simpler choice is to just take \Qref{lem:flatness} also in the essentially unitary case, i.e. we call essentially unitary operators $W$ and $W'$ in \eqref{eq:timeframes} essentially flatband if the chiral blocks of the corresponding half-step operator are essentially unitary up to a factor of $1/\sqrt2$. Moreover, for the sake of brevity we occasionally call a half-step operator $F$ essentially flatband whenever the corresponding $W$ and $W'$ are.

Using the standard form introduced above, we now construct a flatband deformation of walks of the form \eqref{eq:timeframes} directly on the level of the half-step walk $F$:
\begin{lem}\label{lem:flattenF}
	Let $F$ be a
	half-step walk.
	Then there is a continuous path $t\mapsto F_t$, $t\in[0,1]$, such that $F_0=F$ and $F_1$ constitutes essentially flatband walks $W_1$ and $W_1'$.
\end{lem}

\begin{proof}
	By \Qref{lem:flatness}, the essentially flatband condition on timeframed walks is equivalent to essential unitarity of the chiral blocks of the corresponding half-step walk (up to a factor of $1/\sqrt2$). In order to achieve this essential unitarity, we first restrict considerations to the part of the half-step walk where each block is invertible, i.e. to the complement of the kernels and cokernels of the chiral blocks. In the standard form \eqref{eq:Fstructure} this means that we only have to deal with the rightmost block. This reduces the flattening to the task of transforming the invertible blocks to exactly unitary operators up to a factor of $1/\sqrt2$, without destroying the unitarity of the half-step walk along the way.

	Using the polar decomposition, the upper left chiral block $A$ of $F$ can be written as $A=U_A|A|$, with a unique unitary polar isometry $U_A$ and the absolute value $|A|=\sqrt{A^*A}$.
	We can also write the remaining blocks in this way and, moreover, using the unitarity conditions of $F$ in \eqref{eq:unitarity_of_F} we can express all absolute values in terms of $A^*A$. This gives
	\begin{align}
		A&=U_A\sqrt{A^*A}
		&&B=U_B\sqrt{\idty-U_B^*U_AA^*AU_A^*U_B}\\
		C&=U_C\sqrt{\idty-A^*A}
		&&D=U_D\sqrt{U_B^*U_AA^*AU_A^*U_B},
	\end{align}
	where we used $AA^*+BB^*=\idty$ and $XX^*=U_XX^*XU_X^*$ for $B$, $A^*A+C^*C=\idty$ for $C$ and $D^*D+B^*B=\idty$ for $D$. The remaining unitarity conditions all boil down to $U_D=-U_CU_A^*U_B$, which guarantees the unitarity of the flatband half-step walk
	\begin{equation}\label{eq:Fflat}
		F^\flat=\begin{pmatrix}
			U_A & U_B\\U_C & U_D
		\end{pmatrix}/\sqrt2.
	\end{equation}
	In particular, the unitarity of $F$ is independent of $A^*A$, given that $0<A^*A<\idty$, which holds because $A$ is invertible, $A^*A+C^*C=\idty$ and $C^*C>0$. Hence, we can deform $F$ into $F^\flat$ by constructing a continuous path between $A^*A$ and $\idty/2$. A particularly simple path is given by linearly interpolating between the two operators $t\mapsto t(\idty/2)+(1-t)A^*A$.
	
	Taking into account the remaining summands in \eqref{eq:Fstructure}, note that the two left summands of $F$ are already in flat-band form, since their blocks are either zero or unitary. Hence, the above construction works also in this case, with the only difference that the polar isometries of $A,B,C,D$ now might have finite dimensional kernels and cokernels. These, however, only constitute the finite dimensional $\pm1$-eigenspaces of $W$ and $W'$ and are left invariant by construction.
	
	So far, we did not address the essential locality of the half-step walk. However, the path $F_t$ is constructed in a way that respects this property, given that $F$ is essentially local.
	Indeed, by \Qref{lem:esslocpolarisom}, the result of the flattening construction above is essentially local. Moreover, since we only continuously deformed $|A|$, every $F_t$ is essentially local.
\end{proof}

This does not only proof the existence of a flattening path for every half-step walk: It does so by explicitly describing the construction of the flatband half-step walk $F^\flat$ through simply replacing the chiral blocks by their polar isometries.

\subsection{Completeness}\label{sec:completeness_proof}
Before we finally address completeness for the classification of $F$, we state the following theorem which will be important in the proof. It states that for essentially local unitaries the right Fredholm index $\fredP\cdot$ is complete, and thereby extends the analogous result for strictly local walks which was proved in \cite{old_index_paper}. The index $\fredP\cdot$ was generalized to essentially local operators in \cite{LongVersion}. However, the question whether $\fredP\cdot$ is complete for this larger set of essentially local operators was not considered.

\begin{thm}\label{thm:oldindex-complete}
	Let $U$ and $V$ be essentially local unitaries on the one-dimensional lattice. Then $U$ and $V$ are homotopic along an essentially local path if and only if $\fredP{U}=\fredP{V}$.
\end{thm}

\begin{proof}
	Let $\fredP{U}=n$. Then $U=S^{-n}(S^nU)$ and $\fredP{S^nU}=0$, where $S^n$ denotes the shift by $n$ cites. By \Qref{prop:inx0}, $S^nU$ is homotopic to $\idty$ and hence, $U$ is homotopic to $S^{-n}$. This is true for both, $U$ and $V$, and hence we can continuously connect them, via $S^{-n}$.
\end{proof}

The multiplication with shift operators is a standard technique which we also use to alter the Fredholm indices of the chiral blocks of half-step walks in the proof of completeness below. In slight abuse of notation we write $S$ also for the conditional shift on $\HH$ of the form \eqref{eq:H} which shifts a one-dimensional subspace of each $\HH_x$ to the right and leaves the complement invariant.

Having established completeness for the right Fredholm index $\fredP{\cdot}$ for
essentially local unitary operators, allows us to approach the proof of completeness of the independent set of indices $\bigl(\fred{A},\allowbreak \fred{B},\allowbreak \fredP{A},\allowbreak \fredP{B},\allowbreak \fredP{C}\bigr)$ of half-step walks. We start with showing that all half-step walks $F$ with trivial indices are connected to the reference operator
\begin{equation}\label{eq:referenceF}
	F_0=\frac{1}{\sqrt2}\begin{pmatrix}
		\idty & \idty\\ \idty & -\idty
	\end{pmatrix}.
\end{equation}

\begin{lem}\label{lem:trivialComponent}
	Let $F$ be a half-step walk
	with trivial indices, i.e. $\bigl(\fred{A},\fred{B},\fredP{A},\fredP{B},\fredP{C}\bigr)=0\in{\Ir^5}$. Then $F$ is homotopic to $F_0$ in \eqref{eq:referenceF}.
\end{lem}

\begin{proof}
	Consider $F$ in the standard form \eqref{eq:Fstructure}. We begin with modifying the first direct summand in
	\eqref{eq:Fstructure} which consists of two finite-dimensional unitaries $B^-$ and $C^-$ on the off-diagonals. In general, these unitaries are not necessarily of the same size. However, $\fred{A}=0$ implies $\dim\ker A=\dim\ker A^*=\dim\ker D$, where the last equality follows from \eqref{eq:dimkers}. Hence, $B^-$ and $C^-$ act on equivalent finite
	dimensional Hilbert spaces.
	\begin{equation}
		\begin{pmatrix}0&B^-\\C^-&0\end{pmatrix}=
		\begin{pmatrix}B^-&0\\0&C^-\end{pmatrix}
		\begin{pmatrix}0&\idty\\\idty&0\end{pmatrix}
		\quad\sim\quad\begin{pmatrix}B^-&0\\0&C^-\end{pmatrix}
		\begin{pmatrix}\idty
			& \idty\\ \idty & -\idty\end{pmatrix}/\sqrt2=\begin{pmatrix}B^-
			&B^-\\ C^- & -C^-\end{pmatrix}/\sqrt2,
	\end{equation}
	where ``$\sim$'' indicates a homotopy. Using $\fred B=0$ we treat the second summand analogously. The chiral blocks in the third summand in \eqref{eq:Fstructure} are invertible operators on a infinite dimensional Hilbert space which can be made unitary up to a factor of $1/\sqrt2$ via the flattening procedure in \Qref{lem:flattenF}. Then, we are left with a half-step walk of the form \eqref{eq:Fstructure}, where each summand consists of $4$ blocks which are unitary up to a factor of $1/\sqrt2$. Undoing the rearrangement in \eqref{eq:Hstructure2} gives a half-step walk $F$, whose chiral blocks consist of $1/\sqrt2$ times a unitary, which we again denote by $A,B,C$ and $D$.
	
	Since the modifications so far affected $F$ only via a finite rank perturbations and homotopies, its indices are unchanged. We can therefore deform $A,B$ and $C$ to $\idty/\sqrt2$ according to \Qref{thm:oldindex-complete}. $D$ is
	automatically taken care of, by keeping $D=-2CA^*B$ along the path, which guarantees unitarity.
\end{proof}

Having constructed the homotopy to $F_0$ for every $F$ with trivial indices $0\in\Ir^5$, allows us to assemble the proof of the completeness result for the classification of half-step walks $F$ in terms of the indices $\bigl(\fred{A},\fred{B},\fredP{A},\fredP{B},\fredP{C}\bigr)$.
We do this two steps: First we assume both $F_i$ to correspond to gapped walks $W_i$ and $W_i'$ for $i=1,2$, respectively. Then, by \Qref{lem:sixFeu}, $A_i,B_i,C_i$ and $D_i$ are invertible with trivial $\fred\cdot$ indices.
Similarly to the proof of \Qref{thm:oldindex-complete}, we can multiply the $F_i$ with appropriate shift combinations from the left and the right such that the $\fredP\cdot$ indices of the resulting $\widetilde F_i$ vanish. Explicitly, let $\{a,b,c\}=\{\fredP{A},\fredP{B},\fredP{C}\}$ and
consider
\begin{equation}\label{eq:multiplyshifts}
	\widetilde F_i=\begin{pmatrix}
		S^{a} & 0\\ 0 & S^{c}
	\end{pmatrix}
	\begin{pmatrix}
		A_i&B_i\\C_i&D_i
	\end{pmatrix}\begin{pmatrix}
		\idty&0\\0&S^{b-a}
	\end{pmatrix}.
\end{equation}
By $\fredP{{S^{\fredP{X}}X}}=0$ we get $\fredP{\widetilde A_i}=\fredP{\widetilde B_i}=\fredP{\widetilde C_i}=0$
for $i=1,2$. Hence, by \Qref{lem:trivialComponent}, $\widetilde F_i,i=1,2$ are homotopic to $F_0$. Undoing the manipulation \eqref{eq:multiplyshifts} after deforming to $F_0$ by multiplying with the respective inverses we
constructed homotopies of $F_1$ and $F_2$ to the same operator
\begin{equation}\label{eq:F-shift-content}
	F=\begin{pmatrix}
		S^{-a}&S^{-b}\\S^{-c}& -S^{a-b-c}
	\end{pmatrix}/\sqrt2.
\end{equation}

In the general problem where $A_i,B_i,C_i$ and $D_i$ are merely Fredholm we have to take into account the first and the second summands of $F_1$ and $F_2$ in the standard form \eqref{eq:Fstructure}. To the third summands of $F_1$ and $F_2$ we can apply the modification \eqref{eq:multiplyshifts} and thereby assume the summands of the resulting $\widetilde F_1$ and $\widetilde F_2$ to be equal to $F_0$. If we choose the bases for the finite dimensional kernels and co-kernels of $A_i,B_i,C_i$ and $D_i$  appropriately, the $\widetilde F_i$ then take the form

\begin{equation}
	\widetilde F_i=
	\begin{pmatrix}0 & \idty_{d(D_i)}\\ \idty_{d(A_i)} & 0\end{pmatrix}
	\oplus
	\begin{pmatrix}\idty_{d(C_i)} & 0\\ 0 & \idty_{d(B_i)}\end{pmatrix}
	\oplus
	1/\sqrt2\begin{pmatrix}\idty & \idty\\\idty &-\idty\end{pmatrix},
\end{equation}
where $d(X)=\dim\ker X$. Since $\fred{A_i}=d(C_i)-d(B_i)$ and $\fred{A_1}=\fred{A_2}$, $d(C_i)$
and $d(B_i)$ differ by the same amount for both $i=1,2$. Without loss of generality, let $d(C_2)-d(C_1)=d(B_2)-d(B_1)=n>0$. Then we can ``extract'' a finite dimensional block
\begin{equation}
	1/\sqrt2\begin{pmatrix}
		\idty_n &\idty_n\\\idty_n&-\idty_n
	\end{pmatrix}
\end{equation}
from the right summand of $\widetilde F_1$, continuously deform it to the identity and associate it with the middle summand. This enlarges the dimensions of $\ker C_1$ and $\ker B_1$ by $n$ such that $d(C_1)=d(C_2)$ and $d(B_1)=d(B_2)$. Since we only changed $\widetilde F_1$ on a finite dimensional subspace, this perturbation is essentially local. After rearranging the matrix blocks appropriately, the two middle $2\times 2$ matrix blocks of $\widetilde F_1$ and $\widetilde F_2$ coincide. The same procedure applies to the left summand in the standard form and leads to $d(A_1)=d(A_2)$ and $d(D_1)=d(D_2)$ by deforming a finite dimensional block
from the right infinite dimensional summand to $\sigma_x\otimes\idty_m$ with the appropriately
chosen $m$.

This construction continuously connects $\widetilde F_1$ and $\widetilde F_2$
in the set of essentially local half-step walks, and by undoing the left and right multiplication in \eqref{eq:multiplyshifts} on the right summand we indeed obtain $F_1\sim F_2$.

\section{Connection to the classification of timeframed walks}\label{sec:time_framed_walks}
\subsection{Index connections}

Having identified a complete set of indices for half-step walks, we investigate their connection to the walk indices of the two timeframes $W$ and $W'$. Without further restrictions each of these walks is completely characterized by the index set $(\six,\sixR,\six_+)$, respectively $(\six',\sixR',\six'_+)$ \cite{LongVersion}. Clearly, the walk indices are not in one-to-one correspondence with the index set for $F$, since $W$ and $W'$ are not independent. For example, \eqref{pmindices} implies that $\six_\pm(W)=\pm\six_\pm(W')$, which is equivalent to $\six_+(W)=(\six(W)+\six(W')/2$. Hence, also $\six_+(W)$ itself is no longer independent, but determined by the other walk indices (compare also \Qref{thm:twotimeframes}). Apart from that restriction, however, the walk indices are independent, as the following result shows.

\begin{cor}\label{cor:(si+si')/2}
	Two index triples $(\six,\sixR,\six_+)$ and $(\six',\sixR',\six_+')$ have
	representatives $W$ and $W'$ according to \eqref{eq:timeframes} if and only if
	\begin{equation}\label{eq:timeframecondition}
	\six_+(W)=\six_+(W')=\frac{\six(W)+\six(W')}{2}.
	\end{equation}
\end{cor}
\begin{proof}
	The ``only if'' part of the statement follows from \Qref{thm:twotimeframes}.
	For the converse direction let $(\six,\sixR,\six_+)$ and $(\six',\sixR',\six_+')$ be two index triples and assume \eqref{eq:timeframecondition} to hold. Then the generating example defined in \eqref{eq:generating_example} with parameters
    \begin{equation}\label{eq:generatingexample_to_walkindices}
        (n_L,m_L,n_R,m_R,k)=\left(n,\sixR-\six+n,\frac{\six'-\six}{2}+n,\frac{\six'-\six}{2}+\sixR+n,\six-\six'-\sixR+\sixR{}'-2n\right)
	\end{equation}
	gives $W$ and $W'$ with the desired indices where $n\in\Ir$ is arbitrary. By \eqref{eq:timeframecondition} $\six+\six'$ is even, such that either $\six$ and $\six'$ are both even or both odd. In any case, they differ by an even number which guarantees that $(\six'-\six)/{2}=\six'-\six'_+$ in the parameter choice \eqref{eq:generatingexample_to_walkindices} is an integer.
\end{proof}

By this result $\six_+=\six_+'$ are determined by $\six$ and $\six'$ and, moreover, the latter two differ by an even number. Hence, there are only four independent indices in terms of $W$ and $W'$, with the further restriction, that $\six'=2l-\six$, for $l=\six_+\in\Ir$, whereas there are five independent indices for $F$. This poses the question to what extend the classification of half-step walks in terms of $\bigl(\fred{A},\allowbreak \fred{B},\allowbreak \fredP{A},\allowbreak \fredP{B},\allowbreak \fredP{C}\bigr)$ is finer than that of timeframed walks in terms of $\bigl(\six(W),\sixR(W),\six(W'),\sixR(W')\bigr)$.

Before we answer this question, consider the following example:

\begin{example}
    A modification of a half-step walk $F$ which does not change $W$ is to multiply $F$ with a unitary operator which is diagonal in the chiral eigenbasis (see \Qref{lem:Forbit}).
    Consider for example the modification
    \begin{equation}\label{eq:Fmod}
        F\mapsto\begin{pmatrix}S^m & 0\\0 & S^n\end{pmatrix}F,\qquad m,n\in\Ir,
    \end{equation}
    which leads to $\bigl(\fredP{A},\fredP{B},\fredP{C}\bigr)\mapsto\bigl(\fredP{A}-m,\fredP{B}-m,\fredP{C}-n\bigr)$ and does not change the $\fred{\cdot}$-indices. Some choices of the parameters $m,n$ can be detected by the walk indices through
    \begin{equation}\label{eq:sixR_diff}
        \sixR(W)-\sixR(W')=2\fredP{A}-\fredP{C}-\fredP{B}=2\fredP{A}-\fredP{F}.
    \end{equation}
    while others cannot: Choosing $m=n=\ell$ does not change $\sixR-\sixR{}'$, while $m=2\ell,n=0$ induces a change of $-2\ell$. Yet, by \Qref{prop:inx0} both choices imply $\fredP F\mapsto \fredP F-2\ell$.

    This observation has an important consequence: On the level of the walks $W$ and $W'$ it is not always possible to decide whether they stem from a periodically driven continuous time evolution or not. If the time evolution is continuous, i.e. if $F$ is the time-ordered exponential of some Hamiltonian $H(s)$, the half-step operator has to satisfy $\fredP{F}=0$, by item $(3)$ in \Qref{prop:inx0}.
\end{example}

The example suggests that $\fredP{F}$ might serve as the missing index, which is indeed the case:
\begin{cor}\label{cor:completewalkindices}
	Let $W$ and $W'$ be timeframed quantum walks as in \eqref{eq:timeframes} with the half-step walk $F$. Then $\bigl(\six(W),\sixR(W),\six(W'),\sixR(W'),\fredP{F}\bigr)$ is a complete set of indices.
\end{cor}

Note, that differently from the index set \eqref{eq:indexset}, the walk indices, together with $\fredP{F}$ are not pairwise independent, since for example $\six(W)$ and $\six(W')$ have to differ by an even number by \Qref{cor:(si+si')/2}.

\begin{proof}
    To prove this, we relate the index set above to the set of independent indices in \eqref{eq:indexset}, which we already know to be complete from \Qref{thm:complete}. From \Qref{sec:halfstepwalks} we know, how to obtain the walk indices and $\fredP{F}$ from the indices of the corresponding half-step walk. Since the latter are complete, we can also already conclude, that every valid combination of the former can be reached. Hence, it remains to verify, that the converse is also true. From \Qref{lem:sixFeu} we get
    \begin{equation}\label{eq:fredAfredB}
	   \fred{A}=\frac{\six(W)-\six(W')}{2}\qquad\text{and}\qquad\fred{B}=-\frac{\six(W)+\six(W')}{2}.
    \end{equation}
    Since by \Qref{cor:(si+si')/2} $\six(W)$ and $\six(W')$ differ by an even number, the right hand sides are always integers. Moreover, every combination $\{\fred{A},\fred{B}\}\in\Ir^2$ can be reached this way. Further, from \eqref{eq:sixR_diff} and similar expressions we get
    \begin{align}
	   \fredP{A} &= \frac12\fredP{F}+\frac{\sixR(W)-\sixR(W')}{2}\\
	   \fredP{B} &= \frac12\fredP{F}-\frac{\sixR(W)+\sixR(W')}{2}\\
	   \fredP{C} &= \frac12\fredP{F}+\frac{\sixR(W)+\sixR(W')}{2}.
    \end{align}
    Again, by the relation between the combinations of $\sixR(W)$ and $\sixR(W')$ and $\fredP{F}$, these are always integers. And one can again check, that every combination $\bigl(\fredP{A},\fredP{B},\fredP{C}\bigr)\in\Ir^3$ can be realized.
	
    Hence, there is a one to one correspondence between the complete index set \eqref{eq:indexset} and the valid index combinations of walk indices, together with $\fredP{F}$.
\end{proof}

\begin{remark}\label{rm:TransInv}
  If $W$ and $W'$ are translation invariant timeframed walks as in \eqref{eq:timeframes} the situation simplifies: $\six(W)$ and $\six(W')$ vanish identically while $\sixR(W)$ and $\sixR(W')$ are complete invariants for the respective timeframes \cite{UsOnTI}. Via \eqref{eq:fredAfredB} this immediately implies that the corresponding half-step walk is classified by
  \begin{equation}
    \bigl(0,0,\fredP{A},\fredP{B},\fredP{C}\bigr).
  \end{equation}
  By standard arguments, these remaining indices can be calculated as winding numbers \cite{Gohberg,Avron2001}.

  Vice versa, for any two translation invariant chiral symmetric walks $W$ and $W'$ the index condition \eqref{eq:indexcondition} in \Qref{thm:twotimeframes} is trivially fulfilled such that $W$ and $W'$ are timeframes of each other if and only if they are unitarily equivalent for a unitary which commutes with the chiral symmetry on the whole Hilbert space. Inspecting \eqref{eq:Ftwiddle} we find that in this case the half-step walk constructed in the proof of \Qref{thm:twotimeframes} is itself translation invariant.
\end{remark}

\subsection{Bulk-edge correspondence in the timeframed setting}

Bulk-edge correspondence for quantum walks was rigorously proved in the setting with one timeframe. More concretely, it was showed in \cite[Corollary IV.3]{LongVersion} that for a walk $W$ which coincides with bulk walks $W_L$ and $W_R$ far to the left and far to the right, respectively,
\begin{equation}\label{eq:bulk-edge}
\six(W)=\sixR(W_R)+\sixL(W_L)=\sixR(W_R)-\sixR(W_L),
\end{equation}
where for the second equality the bulks have to have proper gaps. Thus, whenever $W_L$ and $W_R$ are in different topological phases, \eqref{eq:bulk-edge} gives a lower bound on the number of symmetry protected eigenvalues of the crossover $W$. However, the theory in \cite{LongVersion} does not predict whether these eigenvalues are at $+1$ or $-1$: This depends on $\six_-(W)$ which cannot be inferred from the asymptotic indices in \eqref{eq:bulk-edge} and, moreover, can be changed by non-gentle perturbations.

Taking into account the second timeframe we consider a situation with two chiral symmetric timeframed crossovers $W$ and $W'$ which correspond to $W_R$ and $W_R'$ far to the right, respectively, and to $W_L$ and $W_L'$ far to the left.
Then, plugging \eqref{eq:bulk-edge} for $W$ and $W'$ into \eqref{eq:timeframecondition} we obtain
\begin{align}
2\six_+(W)&=\phantom{-}2\six_+(W') =\sixR(W_R)+\sixR(W_R')-\bigl(\sixR(W_L)+\sixR(W_L')\bigr) \\
\intertext{and}
2\six_-(W)&=-2\six_-(W')  =\sixR(W_R)-\sixR(W_R')-\bigl(\sixR(W_L)-\sixR(W_L')\bigr).
\end{align}
Thus, the second timeframe stabilizes the symmetry protected eigenvalues of the crossover and allows to attribute them to $+1$ and $-1$. These formulas agree with those in \cite{Asbo2}.

If additionally $W_R,W_R'$ and $W_L,W_L'$ are timeframed with $F_R$ and $F_L$, respectively, the above implies
\begin{align}
\six_+(W)&=\fredP{C_R}-\fredP{C_L}=-\bigl(\fredP{B_R}-\fredP{B_L}\bigr)\\
\intertext{and}
\six_-(W)&=\fredP{A_R}-\fredP{A_L}=-\bigl(\fredP{D_R}-\fredP{D_L}\bigr),
\end{align}
where we used \eqref{eq:sixR}, \eqref{eq:sixRprime} and \eqref{eq:Findex} in combination with $\fredP{F_R}=\fredP{F_L}$, which is a necessary condition for the existence of a unitary crossover between $F_R$ and $F_L$ \cite{old_index_paper}.

The above formulas confirms those in \cite{Asbo3}, where the systems under considerations were given by continuously driven Floquet time evolutions.
Moreover, we generalized them to systems with $\fredP{F}\neq 0$, which indirectly validates \cite{obuse2015unveiling}, where, in the appendix, the invariants from \cite{Asbo3} were applied to systems with $\fredP{F}\neq0$. The results above are also in line with those in \cite{Asbo2}, where a walk was defined by a fixed shift-coin skeleton.

\subsection{Example: The classification of timeframed split-step walks}\label{sec:splitstep}

The stability of the independent set of indices for half-step operators in \Qref{thm:complete} has an important consequence for the classification of the \textbf{split-step walk} model which is defined on $\ell_2(\Ir)\otimes\Cx^2$ as
\begin{equation}\label{eq:splitstep}
W(\theta_1,\theta_2)=R(\theta_1/2)S_\downarrow R(\theta_2)S_\uparrow R(\theta_1/2).
\end{equation}
Here, $R(\theta)=\idty\otimes R_2(\theta)$ rotates the internal degree of freedom homogeneously by the angle $\theta$ around the $y$-axis, and $S_\uparrow$ is the right shift of the spin-up vectors whereas $S_\downarrow$ shifts spin-down vectors to the left. It is straightforward to see that the split-step walk is chiral symmetric with $\gamma=\idty\otimes\sigma_1$.

\begin{figure}
	\begin{tikzpicture}
[
scale=1.5,
font=\small
]
\def\x{1.5}

\definecolor{myblue}{RGB}{100,100,220}
\definecolor{myred}{RGB}{255,100,100}
\definecolor{mygreen}{RGB}{119,221,119}
\definecolor{myyellow}{RGB}{253,253,150}
\definecolor{myorange}{RGB}{255,200,100}


\fill[myorange] 	(-.5*\x,-.5*\x) -- +(-.5*\x,-.5*\x) -- +(-.5*\x,.5*\x)
(-.5*\x,-.5*\x) -- ++(.5*\x,.5*\x) -- ++(.5*\x,-.5*\x) -- +(-.5*\x,-.5*\x)
(.5*\x,-.5*\x) -- +(.5*\x,.5*\x) -- +(.5*\x,-.5*\x);

\fill[myred] 	(-.5*\x,-.5*\x) -- +(-.5*\x,-.5*\x) -- +(.5*\x,-.5*\x)
(-.5*\x,-.5*\x) -- ++(.5*\x,.5*\x) -- ++(-.5*\x,.5*\x) -- +(-.5*\x,-.5*\x)
(-.5*\x,.5*\x) -- +(.5*\x,.5*\x) -- +(-.5*\x,.5*\x);

\fill[mygreen] 	(-.5*\x,.5*\x) -- +(-.5*\x,-.5*\x) -- +(-.5*\x,.5*\x)
(-.5*\x,.5*\x) -- ++(.5*\x,.5*\x) -- ++(.5*\x,-.5*\x) -- +(-.5*\x,-.5*\x)
(.5*\x,.5*\x) -- +(.5*\x,.5*\x) -- +(.5*\x,-.5*\x);

\fill[myblue] 	(.5*\x,-.5*\x) -- +(-.5*\x,-.5*\x) -- +(.5*\x,-.5*\x)
(.5*\x,-.5*\x) -- ++(.5*\x,.5*\x) -- ++(-.5*\x,.5*\x) -- +(-.5*\x,-.5*\x)
(.5*\x,.5*\x) -- +(.5*\x,.5*\x) -- +(-.5*\x,.5*\x);

\draw[white,very thick] (-1*\x,-1*\x) -- (1*\x,1*\x);
\draw[white,very thick] (1*\x,-1*\x) -- (-1*\x,1*\x);
\draw[white,very thick] (-1*\x,0*\x) -- (0*\x,1*\x) (0*\x,-1*\x) -- (1*\x,0*\x);
\draw[white,very thick] (-1*\x,0*\x) -- (0*\x,-1*\x) (0*\x,1*\x) -- (1*\x,0*\x);

\draw (0*\x,1.05*\x) node[above,align=center]{ $\theta_1$};
\draw (1.05*\x,0*\x) node[right,align=left]{ $\theta_2$};

\foreach \i in {-1*\x,0,1*\x}{
	\draw[align=left] (-1.02*\x,{\i}) -- (-1.08*\x,{\i});
	\draw[align=left] ({\i},-1.02*\x) -- ({\i},-1.08*\x);
}

\draw (-1.05*\x,-1*\x)  node[left,align=left]{ $-\pi$};
\draw (-1.05*\x,0*\x)  node[left,align=left]{ $0$};
\draw (-1.05*\x,1*\x)  node[left,align=left]{ $\pi$};

\draw (-1*\x,-1.06*\x)  node[below,align=center]{ $-\pi$};
\draw (0*\x,-1.06*\x)  node[below,align=center]{ $0$};
\draw (1*\x,-1.06*\x)  node[below,align=center]{ $\pi$};

\draw (.5*\x,0*\x) node[align=center]{$\bf{\RNum{4}}$};
\draw (-.5*\x,0*\x) node[align=center]{$\bf{\RNum{3}}$};
\draw (0*\x,-.5*\x) node[align=center]{$\bf{\RNum{2}}$};
\draw (0*\x,.5*\x) node[align=center]{$\bf{\RNum{1}}$};

\def\legendwidth{1.75}
\def\legendhight{0.32}
\def\legendposx{2.7}
\def\legendposy{.2}
\def\circlerad{1ex}
\def\circleshift{-0.5ex}
\newcommand\Square[1]{+(-#1,-#1) rectangle +(#1,#1)}
\def\colorsquare#1{\tikz[baseline=\circleshift]\draw[#1,fill=#1,rotate=45] (0,0) \Square{\circlerad};}

\node[anchor=west] at (\legendposx,\legendposy) {\begin{tabular}{r|  c |c|c}
	 & $\bf{\left(\fredP A,\fredP B,\fredP C\right)}$& &  ${(\sixR(W),\sixR(W'))}$\\\hline
	$\bf{\RNum{1}}$ & $\bf{\left(-1,-1,\hphantom-0\right)}$ & \colorsquare{mygreen}& $\bf (\hphantom-0,\hphantom-1)$\\
	$\bf{\RNum{2}}$ & $\bf{\left(\hphantom-0,\hphantom-0,-1\right)}$ & \colorsquare{myorange} & $\bf (\hphantom-0,-1)$\\
	$\bf{\RNum{3}}$ & $\bf{\left(-1,\hphantom-0,-1\right)}$ & \colorsquare{myred} & $\bf (-1,\hphantom-0)$\\
	$\bf{\RNum{4}}$ & $\bf{\left(\hphantom-0,-1,\hphantom-0\right)}$ & \colorsquare{myblue}& $\bf (\hphantom-1,\hphantom-0)$
	\end{tabular}};

\end{tikzpicture}
	\caption{\label{fig:harlequin} Phase diagram for the half-step operator \eqref{eq:splitstep-F}. For the split-step walk we have $\fred A=\fred B=0$ for all parameters, wherefore, the legend only lists the three non-constant indices $\left(\fredP A,\fredP B,\fredP C\right)$. Through the constraint $\fredP F=\fredP B+\fredP C=-1$ these three indices determine the two indices $\sixR$ and $\sixR'$.}
\end{figure}

The split-step walk $W(\theta_1,\theta_2)$ is translation invariant and, as such, completely classified by the single symmetry index $\sixR$. It is thus known that in the setting $(W)$ in \Qref{sec:walks_protocols}, i.e. if we restrict considerations to $W$ only and do not take the other timeframe $W'$ into account, continuous paths of walks exist between any two patches of equal indices in the phase diagram. Such \textbf{bridges} must lead over the gap closings encountered in the phase diagram of $W$ and therefore leave the split-step model. One way to construct them is explicitly given in \cite[Sec. 3.1.3]{UsOnTI} by regrouping neighbouring cells, i.e. by considering pairs of neighbouring cells as a single cell, see \Qref{fig:bridges}.

However, in the setting $(F)$ where also the other timeframe $W'$ is taken into account some of these bridges do no longer exist.
In the basis where $\gamma=\idty\otimes\sigma_1$ we can read off the half-step walk of $W(\theta_1,\theta_2)$ directly from \eqref{eq:splitstep} to be $R(\theta_2/2)S_\uparrow R(\theta_1/2)$. To apply the classification of half-step walks we transform this into the chiral eigenbasis which yields
\begin{equation}\label{eq:splitstep-F}
F=R(\theta_2/2-\pi/4)S_\uparrow R(\theta_1/2+\pi/4).
\end{equation}

\begin{cor}\label{cor:noFbridge}
	There are patches of equal indices in the phase diagram of the split-step walk $W$ between which there are no homotopies which leave the gaps of $W$ and $W'$ open.
\end{cor}
\begin{figure}
  \begin{tabular}{cc}
    \begin{tikzpicture}[scale=1,=>stealth]

	\definecolor{myblue}{RGB}{100,100,220}
	\definecolor{myred}{RGB}{255,100,100}
	\definecolor{mygreen}{RGB}{119,221,119}
	\definecolor{myyellow}{RGB}{253,253,150}
	\definecolor{myorange}{RGB}{255,200,100}
	\definecolor{nncol}{RGB}{180,180,180}

    \fill[myred] (-3,-1) -- (-1.5,-.629) -- (-1,-1);	
	\fill[myred] (-1.5,-.629) -- (0,-.25) -- (-.5,.125) -- (-2,-.25);
    \fill[myred] (-.5,.125) -- (-1,.5) -- (1,.5);

    \fill[myblue] (-1,-1) -- (.5,-.629) -- (1,-1);	
	\fill[myblue] (.5,-.629) -- (0,-.25) -- (1.5,.125) -- (2,-.25);
    \fill[myblue] (1.5,.125) -- (1,.5) -- (3,.5);


	\fill[nncol] (-2,-.25) -- (-1,.5) -- (-.5,.125);
	\fill[nncol] (0,-.25) -- (-.5,.125) -- (1,.5) -- (1.5,.125);
	\fill[nncol] (2,-.25) -- (1.5,.125) -- (3,.5);
	
    \fill[nncol] (-3,-1) -- (-1.5,-.629) -- (-2.,-.25);	
    \fill[nncol] (-1,-1) -- (-1.5,-.629) -- (0,-.25) -- (.5,-.629);
    \fill[nncol] (1,-1) -- (.5,-.629) -- (2,-.25);

    \draw[white,very thick] (-3,-1) -- (3,.5);
    \draw[white,very thick] (-1,.5) -- (1,-1);

    \draw[white,very thick] (-3,-1) -- (-1,.5);
    \draw[white,very thick] (3,.5) -- (1,-1);

    \draw[white,very thick] (-1,.5) -- (3,.5);
    \draw[white,very thick] (-3,-1) -- (1,-1);

    \draw[white,very thick] (-2,-.25) -- (1,.5);
    \draw[white,very thick] (-1,-1) -- (2,-.25);

    \draw[white,very thick] (-2,-.25) -- (-1,-1);
    \draw[white,very thick] (2,-.25) -- (1,.5);

    \draw[<->,very thick] (-.5,-.6) .. controls +(90:1) and +(100:1) .. (.5,.1);
    
	\draw (2.2,-.25)  node[right,align=center]{{$\theta_2$}};
	\draw (1,0.7) node[above,align=center]{ {$\theta_1$}};

  \end{tikzpicture}    &   \begin{tikzpicture}[scale=1,=>stealth]

	\definecolor{myblue}{RGB}{100,100,220}
	\definecolor{myred}{RGB}{255,100,100}
	\definecolor{redcross}{RGB}{255,7,58}
	\definecolor{mygreen}{RGB}{119,221,119}
	\definecolor{myyellow}{RGB}{253,253,150}
	\definecolor{myorange}{RGB}{255,200,100}
	\definecolor{nncol}{RGB}{180,180,180}

    \fill[myred] (-3,-1) -- (-1.5,-.629) -- (-1,-1);	
	\fill[myred] (-1.5,-.629) -- (0,-.25) -- (-.5,.125) -- (-2,-.25);
    \fill[myred] (-.5,.125) -- (-1,.5) -- (1,.5);

    \fill[myblue] (-1,-1) -- (.5,-.629) -- (1,-1);	
	\fill[myblue] (.5,-.629) -- (0,-.25) -- (1.5,.125) -- (2,-.25);
    \fill[myblue] (1.5,.125) -- (1,.5) -- (3,.5);

	\fill[mygreen] (-2,-.25) -- (-1,.5) -- (-.5,.125);
	\fill[mygreen] (0,-.25) -- (-.5,.125) -- (1,.5) -- (1.5,.125);
	\fill[mygreen] (2,-.25) -- (1.5,.125) -- (3,.5);

	\fill[myorange] (-3,-1) -- (-1.5,-.629) -- (-2.,-.25);	
	\fill[myorange] (-1,-1) -- (-1.5,-.629) -- (0,-.25) -- (.5,-.629);
	\fill[myorange] (1,-1) -- (.5,-.629) -- (2,-.25);


    \draw[white,very thick] (-3,-1) -- (3,.5);
    \draw[white,very thick] (-1,.5) -- (1,-1);

    \draw[white,very thick] (-3,-1) -- (-1,.5);
    \draw[white,very thick] (3,.5) -- (1,-1);

    \draw[white,very thick] (-1,.5) -- (3,.5);
    \draw[white,very thick] (-3,-1) -- (1,-1);

    \draw[white,very thick] (-2,-.25) -- (1,.5);
    \draw[white,very thick] (-1,-1) -- (2,-.25);

    \draw[white,very thick] (-2,-.25) -- (-1,-1);
    \draw[white,very thick] (2,-.25) -- (1,.5);

    \draw[<->,very thick] (-.5,-.6) .. controls +(90:1) and +(100:1) .. (.5,.1);
	\draw[redcross, line width=4pt] (0.3,0.8) -- +(-0.7,-0.7);
	\draw[redcross, line width=4pt] (-.4,0.8) -- +(0.7,-0.7);

	\draw (2.2,-.25)  node[right,align=center]{{$\theta_2$}};
	\draw (1,0.7) node[above,align=center]{ {$\theta_1$}};

\end{tikzpicture}
  \end{tabular}
  \caption{\label{fig:bridges}While in the timeframe on the left bridges between any two patches with equal indices exist \cite{UsOnTI}, in the other timeframe the phase diagram is rotated by $\pi/2$. Therefore, walks which were in the same phase in the first timeframe are now in different phases which renders the existence of some of these bridges impossible.}
\end{figure}

\begin{proof}
  For definiteness, we consider the same patches of the phase diagram of the split-step walk as in \cite{UsOnTI} where $\sixR(W)=0$. After flattening the band structure of the walks, the parameters determining these patches are $\theta_1=0,\theta_2=\pi/2$ and $\theta_1=0,\theta_2=-\pi/2$, respectively. The half-step walks for these parameter pairs are given by
  \begin{equation}\label{eq:F1F2}
    F_1=\frac{1}{\sqrt2}\begin{pmatrix}
                            S&-S\\1&1
                        \end{pmatrix},\qquad \qquad
    F_2=\frac{1}{\sqrt2}\begin{pmatrix}
                            1&1\\-S&S
                        \end{pmatrix},
  \end{equation}
  with index tuples
  \begin{equation}\label{eq:splitstepindexquintuples}
    F_1: (0,0,-1,-1,0)\qquad
    \qquad F_2:(0,0,0,0,-1).
  \end{equation}
  This directly excludes homotopies between $F_1$ and $F_2$ by \Qref{thm:complete}. Also regrouping the cell-structure like in \cite{UsOnTI} does not help, since all indices of $F$ are Fredholm indices and therefore invariant under such regroupings.
\end{proof}

If only the timeframe $W$ is considered we can adjust the index quintuples in \eqref{eq:splitstepindexquintuples} without changing $W$, e.g. by multiplying $F_1$ from the left as in \eqref{eq:Fmod} with $m=-1,n=1$, see also \Qref{lem:Forbit}. This guarantees the existence of a bridge in the setting $(W)$ directly on the level of half-step walks by \Qref{thm:complete}, with an example given by
\begin{equation}\label{eq:F1F2path}
  t\mapsto\frac1{\sqrt2}\begin{pmatrix} 1 & -e^{it\pi}  \\e^{it\pi} S & S \end{pmatrix}.
\end{equation}

While there are no bridges between the patches in the proof of \Qref{cor:noFbridge}, bridges between patches of the phase diagram with the same index combinations of the corresponding half-step walks can be defined explicitly on the level of the half-step walks.
For example, a bridge between the $F_1$ in \eqref{eq:F1F2} and $F_3$ with $\theta_1=\pi,\theta_2=\pi/2$, i.e.
\begin{equation}
	F_3=\frac{1}{\sqrt2}\begin{pmatrix}
	-S&-S\\1&-1
	\end{pmatrix}
\end{equation}
is given by
\begin{equation}
[0,1]\ni t\mapsto F_t=\frac{1}{\sqrt2}\begin{pmatrix}
e^{i\pi t}S&-S\\1&e^{-i\pi t}
\end{pmatrix}.
\end{equation}

We remark that in this paper we only consider chiral symmetry. The split-step walk from \eqref{eq:splitstep} is additionally admissible for particle-hole symmetry. The path $F_t$ between $F_1$ and $F_3$ above breaks this particle-hole symmetry and so does the path between $F_1$ and $F_2$ in \eqref{eq:F1F2path} which only exists in the scenario with one timeframe anyhow. Taking particle-hole symmetry into account a regrouping of cells similar to that described in \cite{UsOnTI} becomes necessary to find homotopies that respect both symmetries.

\subsection{Compact perturbations of half-step walks}
The classification in \cite{LongVersion} is not only stable against homotopies, but also includes stability against compact perturbations that are not continuously contractible along a symmetric path, so-called non-gentle perturbations. In fact, the index $\six_+$ indicates, whether a compact perturbation can be contracted gently or not. In this section, we investigate how compact perturbations influence the index set for $F$.

\begin{lem}\label{lem:pert}
	Let $F$ be a half-step walk and let $W$ and $W'$ be the corresponding timeframed walks. Then every compact perturbation of $F$ is a gentle perturbation of $W$ and $W'$.
\end{lem}

\begin{proof}
	Clearly, every compact perturbation of $F$ is a compact perturbation of $W$ and $W'$. Moreover, the chiral symmetry of $W$ and $W'$ is preserved by arbitrary compact perturbations of $F$. Since any chiral symmetric compact perturbation leaves the symmetry
	indices $\six(W)$ and $\six(W')$ invariant, also
	$\six_+(W)=\six_+(W')$ is invariant by \Qref{cor:(si+si')/2}. Hence, by \cite[Thm.\ VI.4]{LongVersion} in conjunction with \cite[Lem.\ VI.3]{LongVersion} the induced perturbations of $W$ and $W'$ are gentle.
\end{proof}

The converse of this result has an important consequence regarding non-gentle perturbations of chiral symmetric quantum walks: For every $W=\gamma F^*\gamma F$ we know from \Qref{prop:Fexists} that also every perturbed walk $\widetilde W=VW$ possesses a half-step walk $\widetilde F$. However, if the perturbation is non-gentle, the converse of \Qref{lem:pert} implies that $\widetilde F$ cannot be a compact perturbation of $F$. We illustrate that with a simple example:
\begin{example}
  Consider the split-step walk in \eqref{eq:splitstep} with $(\theta_1,\theta_2)=(0,\pi/2)$, i.e.\ $W=\idty\otimes(-i\sigma_2)$ which is chiral symmetric for $\gamma=\idty\otimes\sigma_1$. Its half-step walk is given by $F=R(\pi/4)S_\uparrow$. A typical non-gentle perturbation of $W$ is to replace the coin at $x=0$ by $\sigma_1$ \cite{LongVersion,splitstepexplorer}. On the level of the half-step walk, this can be achieved, by modifying $S_\uparrow$ on a half-chain: With $\widetilde S_\uparrow=(\Pl 0+\Pg 0\gamma \Pg0 )S_\uparrow(\Pl 0+\Pg 0\gamma \Pg 0)$, $\widetilde F=R(\pi/4)\widetilde S_\uparrow$ is a half-step operator for $\widetilde W$. Clearly, $F-\widetilde F$ is not compact.
\end{example}
For finite systems, e.g. systems with periodic boundary conditions this implies the following: Since in such systems there are no non-compact perturbations, a single non-gentle perturbation renders impossible the existence of a half-step walk $F$. On the other hand, every large-scale modification of $F$ needed to implement a local non-gentle perturbation of $W$ produces another local non-gentle perturbation somewhere else on the ring. These two non-gentle perturbations are globally gentle.

\section{Outlook}
In this work, we have given a complete classification of chiral symmetric half-step walks in terms of their block decomposition with respect to the chiral eigenbasis. The classification is characterized by five indices and we presented a generating example that allows us to construct a representative for each class. We clarified how this classification is connected to the classification of the different timeframes of the half-step walk, and under which conditions a generating continuous-time process exists. We conclude with some remarks on possible extensions of our work.

Within the class of chiral symmetric walks, we chose to consider here, no assumption about translation invariance or other regularity assumptions have been made. As we have argued in Remark \ref{rm:TransInv} the half-step operator itself can be chosen translation invariant if we start from translation invariant timeframes. It would be interesting to investigate how the classification in terms of half-step unitaries presented here will have to be adjusted under such a restricted class of deformations and whether additional obstructions might arise.

A point we touch on only peripherically is the structure of the half-step walks. In many works (see e.g. \cite{Asbo1,Asbo2,Asbo4}) the walk is given explicitly as a ``protocol'' sequence of coin operations, acting separately on each site, and state dependent shifts. When chiral symmetric, such walk protocols have an inversion point and we can think of the evolution up to the inversion point as a concrete realization of a half-step walk. Protocols are necessarily strictly local, i.e. distant matrix elements vanish exactly. Restricting to this class of walks might thus jeopardize the results above. For example it is unclear whether a completeness statement like that in \Qref{thm:complete} holds within the set protocols. The classification of protocols with fixed shift-coin skeletons can be important for experimental implementations where the protocol is fixed by the design of the experiment.

We further leave open the question of a classification in higher lattice dimensions. Naturally, this is a very interesting research direction where also the theory for the one-step scenario is still open. While the results of \Qref{sec:halfstepwalks} are independent of any locality property of the half-step walks and thus apply in any lattice dimension, the main obstacle is to find analogues of the directed Fredholm indices for the chiral blocks introduced in \Qref{sec:locality}.

Similarly, it would be interesting to analyze the consequences of this classification beyond the single particle regime, i.e. to characterize how this classification might put restrictions on the deformation of chiral symmetric quantum cellular automata.

\section*{Acknowledgements}
We thank the referees for helpful comments, which improved the readability of the manuscript. Moreover, we thank Janos Asb\'oth for inspiring discussions and for directing our interest to timeframed quantum walks in the first place.

C. Cedzich acknowledges support by the projet PIA-GDN/QuantEx P163746-484124 and by {\em DGE -- Ministère de l'Industrie}.

T. Geib and R. F. Werner acknowledge support from the DFG through SFB 1227 DQ-mat.

A. H. Werner thanks the VILLUM FONDEN for its support with a Villum Young Investigator Grant (Grant No. 25452) and its support via the QMATH Centre of Excellence (Grant No. 10059). 

\bibliography{F2Wbib}

\end{document}